\documentclass[11pt]{amsart}
\usepackage{fullpage}
\usepackage{latexsym}
\usepackage{amsmath}
\usepackage{amsfonts}
\usepackage{amssymb}
\usepackage{amsthm}
\usepackage{verbatim}
\usepackage{enumitem}
\usepackage{tikz}
\usetikzlibrary{arrows,matrix,decorations,arrows,shapes}
\usepackage[ruled,titlenumbered,linesnumbered,noend,norelsize]{algorithm2e}
\DontPrintSemicolon \SetNlSty{footnotesize}{}{}
\IncMargin{0.5em} \SetNlSkip{1em}
\SetKwIF{If}{ElseIf}{Else}{if}{then}{else if}{else}{endif} 
\usepackage{algorithmic}
\usepackage[foot]{amsaddr}
\usepackage{thm-restate}

\usepackage{marginnote}

\newtheorem{theorem}{Theorem}
\newtheorem{lemma}[theorem]{Lemma}
\newtheorem{claim}{Claim}[theorem]
\newtheorem{prop}[theorem]{Proposition}

\newtheorem{definition}[theorem]{Definition}

\renewcommand{\le}{\leqslant}
\renewcommand{\leq}{\leqslant}
\renewcommand{\ge}{\geqslant}
\renewcommand{\geq}{\geqslant}

\newcommand{\gr}[1]{\mathfrak{#1}}

\DeclareMathOperator{\val}{val}
\DeclareMathOperator{\FF}{FF}
\DeclareMathOperator{\Forb}{Forb}
\DeclareMathOperator{\width}{width}
\DeclareMathOperator{\Min}{Min}
\DeclareMathOperator{\Max}{Max}
\newcommand{\abs}[1]{\left|#1\right|}    

\makeatletter
\providecommand*{\cupdot}{%
  \mathbin{%
    \mathpalette\@cupdot{}%
  }%
}
\newcommand*{\@cupdot}[2]{%
  \ooalign{%
    $\m@th#1\cup$\cr
    \sbox0{$#1\cup$}%
    \dimen@=\ht0 %
    \sbox0{$\m@th#1\cdot$}%
    \advance\dimen@ by -\ht0 %
    \dimen@=.5\dimen@
    \hidewidth\raise\dimen@\box0\hidewidth
  }%
}

\providecommand*{\bigcupdot}{%
  \mathop{%
    \vphantom{\bigcup}%
    \mathpalette\@bigcupdot{}%
  }%
}
\newcommand*{\@bigcupdot}[2]{%
  \ooalign{%
    $\m@th#1\bigcup$\cr
    \sbox0{$#1\bigcup$}%
    \dimen@=\ht0 %
    \advance\dimen@ by -\dp0 %
    \sbox0{\scalebox{2}{$\m@th#1\cdot$}}%
    \advance\dimen@ by -\ht0 %
    \dimen@=.5\dimen@
    \hidewidth\raise\dimen@\box0\hidewidth
  }%
}
\makeatother

\title{An Easy Subexponential Bound for online Chain Partitioning}

\author{Bart\l omiej Bosek$^1$ $^\dagger$}
\address{$^1$Theoretical Computer Science Department,
Faculty of Mathematics and Computer Science,
Jagiellonian University in Krak\'ow, ul. \L{}ajsiewicza 6, Krak\'{o}w 30-348, Poland.}
\email{bosek@tcs.uj.edu.pl}
\thanks{$^\dagger$Research of these authors is supported by Polish National Science Center (NCN) grant 2011/03/B/ST6/01367.}

\author{Hal A. Kierstead$^2$ $^\ddagger$}
\address{$^2$School of Mathematical Sciences and Statistics, Arizona
State University, Tempe, AZ 85287, USA.
}
\email{kierstead@asu.edu}
\thanks{$^\ddagger$Research of this author was supported in part by NSF grant DMS-0901520}

\author{Tomasz Krawczyk$^1$ $^\dagger$}
\email{krawczyk@tcs.uj.edu.pl}

\author{Grzegorz Matecki$^1$ $^\dagger$}
\email{matecki@tcs.uj.edu.pl}

\author{Matthew E. Smith$^2$}
\email{mattearlsmith@gmail.com}

\date{\today}

\begin{document}

\begin{abstract}
Bosek and Krawczyk exhibited an online algorithm for partitioning an online poset of width $w$ into $w^{14\lg w}$ chains.
We improve this to $w^{6.5 \lg w + 7}$ with a simpler and shorter proof by combining the work of
Bosek \& Krawczyk with work of  Kierstead \& Smith on First-Fit chain partitioning of ladder-free posets.
We also provide examples illustrating the limits of our approach.

\end{abstract}


\maketitle

\section{Introduction}
\label{sec:introduction}

An online poset $P^{\prec}$ is a triple $(V,\le_{P}, \prec)$, where $P=(V,\le_{P})$ is a poset and
$\prec$ is a total order on $V$, called the \emph{presentation order} of $P$.
Let $P^{v_i}$ be induced by the first $i$ vertices $v_{1}\prec\dots\prec v_{i}$.
An \emph{online chain partitioning algorithm} is a deterministic algorithm $\mathcal{A}$ that assigns the vertices $v_{1}\prec\dots\prec v_{n}$ of  $P$ to disjoint chains $C_{1},\dots,C_{t}$ so that for each $i$, the chain $C_{j}$ to which $v_{i}$ is assigned, is determined solely by the subposet $P^{v_i}$.
This formalizes the scenario in which the algorithm $\mathcal{A}$ receives the vertices of $P$ one at a time, and when a vertex is received, irrevocably assigns it to one of the chains.
Let $\chi_{\mathcal{A}}(P^{\prec})$ denote the number of (nonempty) chains that $\mathcal{A}$ uses to partition $P^{\prec}$,
and $\chi_{\mathcal{A}}(P)=\max_{\prec}(\chi(P^{\prec}))$ over all presentation orders $\prec$ for $P$.
For a class of posets $\mathcal{P}$, let $\val_{\mathcal{A}}(\mathcal{P})=\max_{P\in\mathcal{P}}(\chi_{\mathcal{A}}(P))$
and $\val(\mathcal{P})=\min_{\mathcal{A}}(\val_{\mathcal{A}}(\mathcal{P}))$ over all online chain partitioning algorithms $\mathcal{A}$.
Our goal is to bound $\val(\mathcal{P}_{w}$), where $\mathcal{P}_{w}$ is the class of finite posets of width $w$ (allowing countably infinite posets with $w$ finite in $\mathcal{P}_{w}$ would not effect  results).

By Dilworth's Theorem \cite{Dil}, every poset with finite width $w$ can be partitioned into $w$ chains, and this is best possible.
However this bound cannot be achieved online. In 1981, Kierstead proved
\begin{theorem}[\cite{K-Dil}]
\label{KThm}
$4w-3\le\val(\mathcal{P}_{w})\le\frac{5^{w}-1}{4}$.
\end{theorem}
Kierstead asked whether $\val(\mathcal{P}_{w})$ is polynomial in $w$, and noted that his methods also provided a super linear lower bound.
Until recently, there was little progress.
Szemer\'edi (see \cite{K-sur}) proved a quadratic lower bound, which was improved to $(2 - o(1)) \binom{w+1}{2}$ by Bosek et al. \cite{BFKKMM}.
In 1997 Felsner \cite{Fel} proved  $\val (\mathcal P_2) \le 5$, and in 2008 Bosek \cite{Bos-th} proved  $\val (\mathcal P_3)\le16$.
In 2010 Bosek and Krawczyk made a major advance by proving a subexponential bound.

\begin{theorem}[\cite{BK-FOCS,BK15}]\label{BKThm}$\val (\mathcal P_w)\le w^{14 \lg w}$.
\end{theorem}

Based on \cite{BK15,KSm} we provide a much  shorter and simpler proof  of a slightly improved bound:
\begin{restatable}{theorem}{valbound}
\label{valbound}%
 $\val(\mathcal P_w) \le w^{6.5 \lg w + 7}.$
\end{restatable}

The difference between the proof of Theorem \ref{KThm} and the proofs of Theorems
\ref{BKThm} and \ref{valbound} is fundamental. In the former relations are \emph{added}
to the online poset $P^{\prec}$ to create a new online poset $Q^{\prec}$ with smaller
width so that every online chain of $Q$ can be partitioned into $5$ online chains
of $P$; then induction is applied. In the latter relations are \emph{deleted} from
$P^{\prec}$ to form an online poset $Q^{\prec}$ with the same width; this would
seem to make it harder to partition $Q$, but paradoxically limits the wrong choices
an algorithm can make.

The simplest online chain partitioning algorithm is First-Fit, which assigns each new vertex $v_{i}$ to the chain $C_{j}$,
with the least index $j\in\mathbb{Z}^+$ such that for all $h<i$ if $v_{h}\in C_{j}$ then $v_{h}$ is comparable to $v_{i}$.
It was observed in \cite{K-Dil} that $\val_{\FF}(\mathcal{P}_{w})=\infty$
(see \cite{K-sur} for details) for any $w>1$.
The poset used to show this fact contains substructures that are important to this paper, so we present it.

\begin{lemma}[\cite{K-Dil}]\label{CE}
For every $n\in\mathbb Z^+$ there is an online poset $R_n^{\prec}$
with $\width(R_n^{\prec})\le 2$ and $\chi_{\FF}(R_n^{\prec})=n$.
\end{lemma}
\begin{proof}
We define the online poset $R_n^{\prec}=(X,\le_R,\prec)$ as follows.
The poset $R_n$ consists of $n$ chains $X^1,\dots,X^n$ with $$X^{k}= x_k^k \le_Rx_{k-1}^k\le_R \dots \le_R x_2^k \le_R x_1^k$$ and the additional comparabilities and incomparabilities given by:
$$x_i^k \ge_R X^{1} \cup X^{2} \cup \dots \cup X^{k-2} \cup \{ x^{k-1}_{k-1} , x^{k-1}_{k-2} , \ldots , x^{k-1}_{i} \} $$
$$x_i^k \parallel_R \{ x^{k-1}_{i-1} , x^{k-1}_{i-2} , \ldots , x^{k-1}_{1} \}.$$
Note that the superscript of a vertex indicates to which chain  $X^k$ it belongs and the subscript is its index within that chain.
The example of $R_5$ is illustrated in Figure~\ref{phase5}.
The presentation order $\prec$ is given by $X^1\prec \dots \prec X^n$, where the order $\prec$ on the vertices of $X^k$ is the same as $\le_R$ on $X^k$.

Observe that $X^{k-2}\le_R X^k$.
Hence, the width of $R_n$ is $2$.
By induction on $k$ one can show that each vertex $x^k_i$ is assigned to chain $C_i$.
\end{proof}

\begin{figure}[tbh]
\begin{center}
\begin{tikzpicture}
\path (.75,-1) coordinate (t)
node {$R_5$};
\path (0,0) coordinate (v1);
\fill (v1) circle (3pt)
node[left] {$x^1_{1\phantom{1}}$};

\path (1.5,0) coordinate (w2);
\fill (w2) circle (3pt)
node[right] {$\phantom{1} x_2^2$};
\path (1.5,1) coordinate (w1);
\fill (w1) circle (3pt)
node[right] {$\phantom{1} x_1^2$};

\path (0,1) coordinate (x3);
\fill (x3) circle (3pt)
node[left] {$x^3_{3\phantom{1}}$};
\path (0,2) coordinate (x2);
\fill (x2) circle (3pt)
node[left] {$x^3_{2\phantom{1}}$};
\path (0,3) coordinate (x1);
\fill (x1) circle (3pt)
node[left] {$x^3_{1\phantom{1}}$};

\path (1.5,2) coordinate (y4);
\fill (y4) circle (3pt)
node[right] {$\phantom{1} x_4^4$};
\path (1.5,3) coordinate (y3);
\fill (y3) circle (3pt)
node[right] {$\phantom{1} x_3^4$};
\path (1.5,4) coordinate (y2);
\fill (y2) circle (3pt)
node[right] {$\phantom{1} x_2^4$};
\path (1.5,5) coordinate (y1);
\fill (y1) circle (3pt)
node[right] {$\phantom{1} x_1^4$};

\path (0,4) coordinate (z5);
\fill (z5) circle (3pt)
node[left] {$x^5_{5\phantom{1}}$};
\path (0,5) coordinate (z4);
\fill (z4) circle (3pt)
node[left] {$x^5_{4\phantom{1}}$};
\path (0,6) coordinate (z3);
\fill (z3) circle (3pt)
node[left] {$x^5_{3\phantom{1}}$};
\path (0,7) coordinate (z2);
\fill (z2) circle (3pt)
node[left] {$x^5_{2\phantom{1}}$};
\path (0,8) coordinate (z1);
\fill (z1) circle (3pt)
node[left] {$x^5_{1\phantom{1}}$};

\draw (v1) -- (z1);
\draw (w2) -- (y1);

\draw (v1) -- (w1);
\draw (w1) -- (x1);
\draw (x1) -- (y1);
\draw (y1) -- (z1);

\draw (w2) -- (x2);
\draw (x2) -- (y2);
\draw (y2) -- (z2);

\draw (x3) -- (y3);
\draw (y3) -- (z3);

\draw (y4) -- (z4);

\path (5.5,0) coordinate (e)
node {$L_m$};
\path (6,0+1) coordinate (x1);
\fill (x1) circle (3pt)
node[right] {$x_{1}$};
\path (6,1+1) coordinate (x2);
\fill (x2) circle (3pt)
node[right] {$x_{2}$};
\path (6,3.5+1) coordinate (x3);
\fill (x3) circle (3pt)
node[right] {$x_{m-1}$};
\path (6,4.5+1) coordinate (x4);
\fill (x4) circle (3pt)
node[right] {$x_{m}$};

\path (5,1+1) coordinate (y1);
\fill (y1) circle (3pt)
node[left] {$y_{1\phantom{-1}}$};
\path (5,2+1) coordinate (y2);
\fill (y2) circle (3pt)
node[left] {$y_{2\phantom{-1}}$};
\path (5,4.5+1) coordinate (y3);
\fill (y3) circle (3pt)
node[left] {$y_{m-1}$};
\path (5,5.5+1) coordinate (y4);
\fill (y4) circle (3pt)
node[left] {$y_{m\phantom{-1}}$};

\draw (6,0+1) -- (6,1.5+1);
\draw (6,3+1) -- (6,4.5+1);
\draw (5,1+1) -- (5,2.5+1);
\draw (5,4+1) -- (5,5.5+1);

\draw (x1) -- (y1);
\draw (x2) -- (y2);
\draw (x3) -- (y3);
\draw (x4) -- (y4);

\path (6,2+1) coordinate (d1);
\fill (d1) circle (1pt);
\path (6,2.25+1) coordinate (d2);
\fill (d2) circle (1pt);
\path (6,2.5+1) coordinate (d3);
\fill (d3) circle (1pt);

\path (5,3+1) coordinate (d111);
\fill (d111) circle (1pt);
\path (5,3.25+1) coordinate (d222);
\fill (d222) circle (1pt);
\path (5,3.5+1) coordinate (d333);
\fill (d333) circle (1pt);

\end{tikzpicture}
\end{center}
\caption{Hasse diagrams of $R_5$ and $L_m$.}\label{phase5}
\end{figure}
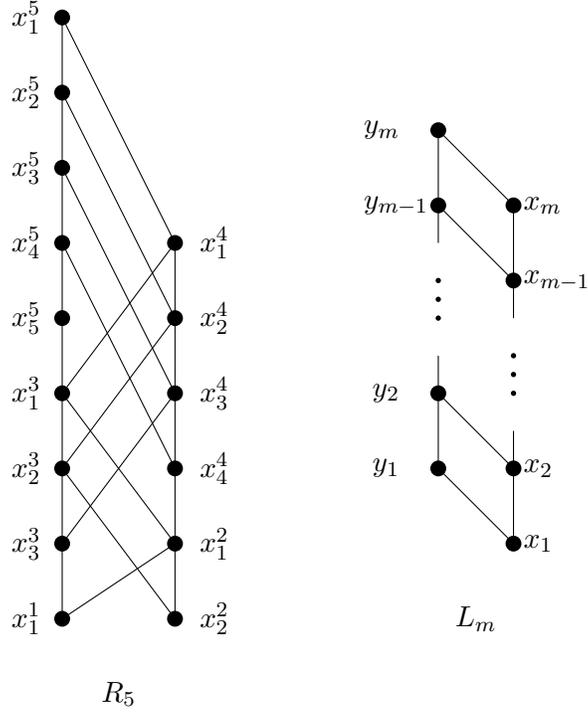

Despite Lemma~\ref{CE}, the analysis of the performance of First-Fit on restricted classes of posets has been useful and interesting.
For posets $P$ and $Q$, we say $P$ is \textit{$Q$-free} if $P$ does not contain $Q$ as an induced subposet.
Let $\Forb(Q)$ denote the family of $Q$-free posets, and $\Forb_w(Q)$ denote the family of $Q$-free posets of width at most $w$. 
Abusing notation, we write $\val_{\FF}(Q,w)$ for $\val_{\FF}(\Forb_w(Q))$.

Let $\mathbf s$ denote the total order (chain) on $s$ vertices, and $\mathbf s + \mathbf t$ denote the width $2$  poset consisting of disjoint copies of $\mathbf s$ and $\mathbf t$ with no additional comparabilities or vertices.
It is well known~\cite{fish} that the class of interval graphs is equal to $\Forb(\mathbf2+\mathbf2)$. First-Fit chain partitioning of interval orders has applications to polynomial time approximation algorithms \cite{K-FF,K2} and Max-Coloring~\cite{PRV}.
The first linear upper bound $\val_{\FF}(\mathbf2+\mathbf2,w)\le40w$ was proved by Kierstead in 1988 \cite{K-FF}.
This was improved later to $\val_{\FF}(\mathbf2+\mathbf2,w) \le 26w$ in~\cite{KQ}.
In 2004 Pemmaraju, Raman, and K. Varadarajan \cite{PRV} introduced a beautiful new technique to show $\val_{\FF}(\mathbf2+\mathbf2,w)\le10w$, and this was quickly improved to $\val_{\FF}(\mathbf2+\mathbf2,w)\le8w$ \cite{BKT,NB}.
In 2010 Kierstead, D. Smith, and Trotter \cite{KST, DS-th} proved
$5(1-o(1))w\le\val_{\FF}(\mathbf2+\mathbf2,w)$.
In 2010 Bosek, Krawczyk, and Szczypka \cite{BKSsiam} proved that $\val_{\FF}(\mathbf t+\mathbf t, w)\le3tw^2$.
This result plays an important role in the proof of Theorem~\ref{BKThm}.
Joret and Milans \cite{JM} improved this to
$\val_{\FF}(\mathbf s+\mathbf t,w)\le 8(s-1)(t-1)w$.
Recently, Dujmovi\'{c}, Joret, and Wood \cite{DJW} proved $\val_{\FF}(\mathbf t+\mathbf t,w)\le16tw$.
In 2010 Bosek, Krawczyk, and Matecki proved:
\begin{theorem}[\cite{BKM}]
For every width $2$ poset $Q$ there is a function $f_Q$ with $\val_{\FF}(Q,w)\le f_Q(w)$.
\end{theorem}
Lemma~\ref{CE} shows that the theorem cannot be extended to posets $Q$ with width greater than $2$.

Let $m\in \mathbb{Z}^+$.
An \emph{$m$-ladder} is a poset  $L_m=L(x_1\dots x_m;y_1\dots y_m)$ with vertices $x_1,y_1,\dots, x_m,y_m$ such that
$x_1 <_L  \dots <_L x_m$, $y_1  <_L \dots <_L y_m$, $x_i <_L y_j$ for $1 \leq i \leq j \leq m$, and $x_i \parallel_L y_{j}$ for $1\leq j < i \leq m$.
The vertices $x_1, \ldots, x_m$ are the \textit{lower leg} and the vertices $y_1, \ldots , y_m$ are the \textit{upper leg} of  $L_m$.
The vertices $x_i$, $y_i$ together form the \textit{$i$-th rung} of $L_m$.
We provide a Hasse diagram of $L_m$ in Figure~\ref{phase5}.
Notice that for two consecutive chains $X^i$ and $X^{i+1}$ of $R_n$, the set $X^i\cup(X^{i+1}-x^{i+1}_{i+1})$ induces the ladder $L_i$ in $R_n$.

Our attack is based on the following observation of Bosek and Krawczyk, first mentioned in \cite{BK-FOCS, BK15}, but never proved so far.
\begin{lemma} \label{ladobs}
$\val(\mathcal P_w)\le w\val_{\FF}(L_{2w^2+1},w)$ for $w \in \mathbb Z^+$.
\end{lemma}
In this paper we provide the first proof of the above-mention lemma.
Kierstead and Smith completed this attack with the next lemma.
\begin{restatable}[\cite{KSm}]{lemma}{ladupbound}
\label{ladupbound}
$\val_{\FF} (L_m, w) \le  w^{2.5 \lg (2w) + 2 \lg m}$ for $m, w \in \mathbb Z^+$.
\end{restatable}
Combining Lemmas~\ref{ladobs} and ~\ref{ladupbound} we get $\val(\mathcal P_w) \le w^{6.5 \lg w + 7}$, which completes the proof of Theorem~\ref{valbound}.
Beside that, the paper presents two new constructions to show that the bounds given in Lemmas~\ref{ladobs} and ~\ref{ladupbound} can not be improved substantially and hence a new technique will be needed to prove a polynomial upper bound on $\val(\mathcal P_w)$.

This paper is organized as follows.
Section~\ref{sec:preliminaries} introduces some notation and definitions.
In Section \ref{sec:online-algorithm} we present our online algorithm
and reduce the proof of its performance bound to proving Lemmas~\ref{ladobs} and ~\ref{ladupbound},
which are shown in Sections~\ref{sec:R-is-ladder-free} and \ref{sec:firstfit}.
In Section \ref{sec:limitations} we present constructions that show limitations of our approach.
Section \ref{sec:concluding} contains some concluding observations.

\section{Preliminaries}
\label{sec:preliminaries}

Let $P = (V,\le_P)$ be a poset with $u,v \in V$.
We usually write $u \in P$ for $u \in V(P)$.
The \textit{upset of $u$ in $P$} is $U_P (u) =\{ v : u <_P v\}$, the \textit{downset of $u$ in $P$} is $D_P (u) =\{ v : v <_P u\}$, and the \textit{incomparability set of $u$ in $P$} is $I_P (u) = \{ v : v \parallel_P u \}$.
The \textit{closed upset} and \textit{closed downset of $u$ in $P$} are, respectively, $U_P [u] = U_P (u) + u$ and $D_P [u] = D_P (u) + u$.
Define $[u,v]_P = U_P [u] \cap D_P [v]$. For $U \subseteq V$, define $D_P (U) = \bigcup_{u \in U} D_P (u)$, $U_P (U) = \bigcup_{u \in U} U_P (u)$, $D_P [U] = D_P (U) \cup U$ and
$U_P [U] = U_P (U) \cup U$. If $U' \subseteq V$, let $[U,U']_P = U_P [U] \cap D_P [U']$.
The \textit{subposet of $P$ induced by $U$} is denoted by $P[U]$, and  $P-u$ denotes $P[V-u]$.
If $U_P (u) = \emptyset$, then $u$ is \textit{maximal}. If $D_P (u) = \emptyset$, then $u$ is \textit{minimal}. If $D_P [u] = P$, then $u$ is \textit{maximum}.
If $U_P [u] = P$, then $u$ is \textit{minimum}.
Let $\Max_P (U)$ be the set of maximal vertices in $P[U]$ and $\Min_P (U)$ be the set of minimal vertices in $P[U]$.
Let $\Max_P=\Max_P (V)$ and $\Min_P=\Min_P (V)$.

A chain partition $\mathcal C$ of $P$ is a \textit{Dilworth partition} if $|\mathcal C|=\width(P)$.
If vertices $u$ and $v$ are in the same chain of some Dilworth partition then  $uv$ is called a \textit{Dilworth edge} of $P$.

Let $\mathcal M_P= (\mathcal V_P, \sqsubseteq_P)$, where $\mathcal V_P$ is the set of maximum antichains in $P$ and $\sqsubseteq_P$  is defined by
$$A \sqsubseteq_P B \text{ if }  A \subseteq D_P [B] \text{ (or equivalently $B \subseteq U_P [A]$)}.$$
If $A \sqsubseteq_P B$ and $A \neq B$, we write $A \sqsubset_P B$.
In \cite{aclat} Dilworth showed that $\mathcal M_P$ is a lattice with the meet and the join defined by
$$A \wedge B  = \Min_P \{ A \cup B \}~\mathrm{and}~
A \vee B   =\Max_P \{A \cup B \}. $$

A poset $P=(V,\le_P)$ is \emph{bipartite} if the set $V$ can be partitioned into two disjoint antichains $A,B$ such that
$A \sqsubset_P B$ --- such a poset is denoted by $(A,B,\le_P)$.
A bipartite poset $P=(A,B,\le_P)$ is a \textit{core} if $|A|=|B|$ and for any comparable pair $x \le_P y$ with $x \in A$ and $y \in B$,
$xy$ is a Dilworth edge (see Figure~\ref{excore}).
Informally, we think of a core as a bipartite poset whose Hasse diagram is a balanced bipartite graph in which each edge is included in some perfect matching.

\begin{figure}[tbh]
\begin{center}
\begin{tikzpicture}
\path (0,0) coordinate (R2);
\path (7,0) coordinate (R1);

%
%

\foreach \i in {1,...,5}{
\path (R2) ++(\i,1) coordinate (y2\i);
\fill (y2\i) circle (3pt);
\path (R2) ++(\i,0) coordinate (x2\i);
\fill (x2\i) circle (3pt);}

\foreach \i in {2,...,5}{
\draw (x2\i) -- (y2\i);
\draw (x21) -- (y2\i);}
\draw (x21) -- (y21);
\draw (x22) -- (y21);
\draw (x23) -- (y22);
\draw (x24) -- (y23);
\draw (x25) -- (y24);

\path (x21) ++ (-1,1/2)node{$Q$};
\path (x25) ++ (1,1/2)node{$\phantom{Q}$};

\foreach \i in {1,...,5}{
\path (R1) ++(\i,1) coordinate (y1\i);
\fill (y1\i) circle (3pt);
\path (R1) ++(\i,0) coordinate (x1\i);
\fill (x1\i) circle (3pt);}

\foreach \i in {1,...,5}{
\draw (x1\i) -- (y1\i);}
\foreach \i in {3,...,5}{
\draw (y15) -- (x1\i);}
\draw (x11) -- (y12);
\draw (x12) -- (y11);
\draw (x12) -- (y13);
\draw (x14) -- (y13);
\draw (x15) -- (y14);

\path (x11) ++ (-1,1/2)node{$R$};
\path (x15) ++ (1,1/2)node{$\phantom{R}$};
\path (x12) ++ (0,-1/3)node{$x$};
\path (y13) ++ (0,1/3)node{$y$};
\end{tikzpicture}
\end{center}
\caption{Poset  $Q$ is a core of width 5. $R$ is not a core since $xy$ is not a Dilworth edge.}\label{excore}
\end{figure}
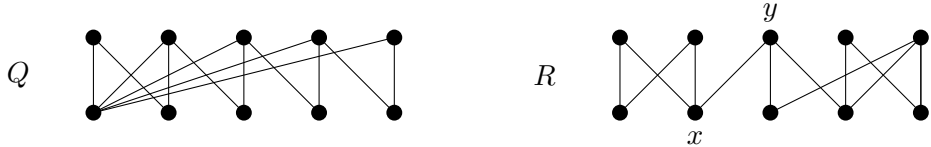

A chain in a poset $P$ corresponds to an independent set in its cocomparability graph.
Offline the terms chain partition and coloring are interchangeable, but an online chain partitioning algorithm has more information to use than an online coloring algorithm.
This advantage is lost by First-Fit.


The notion of Grundy coloring is useful for analyzing First-Fit.

\begin{definition}\label{Pgr}
Let  $n\in \mathbb Z^+$.  A function $\gr{g} : P \rightarrow [n]$ is an \textit{$n$-Grundy coloring of a poset $P$} if
\begin{enumerate}[label=(G\theenumi)]
\item\label{G:1} for each $i \in [n]$, the set $\{ u \in P : \gr{g} (u) = i \}$ is a chain in $P$;
\item\label{G:2} for each $i \in [n]$, there is some $u \in P$ so that $\gr{g} (u) = i$ (i.e.: $\gr{g}$ is surjective); and
\item\label{G:3} if $v \in P$ with $\gr{g} (v) = j$, then for all $i \in [j-1]$ there is some $u \in I_P (v)$ such that $\gr{g}(u) = i$.
\end{enumerate}
\end{definition}

Often, we call the elements of $[n]$ as \textit{colors}.
If $u\in P$ and $\gr{g} (u) = i$, we say $u$ is \textit{colored with}~$i$.   
Let $u,v \in P$.  
If $u \parallel_P v$ and $\gr{g} (u) < \gr{g}(v)$, we say $u$ is a \textit{$\gr{g}(u)$-witness for $v$ under $\gr{g}$}.
%
%

The next lemma is folklore. 
It allows the analysis of a dynamic online process in a static setting.

\begin{lemma} \label{PgrP}
For any poset $P$, the largest $n$ for which $P$ has an $n$-Grundy coloring  is equal to $\chi_{\FF} (P)$.
\end{lemma}

\section{The online algorithm}
\label{sec:online-algorithm}

In this section we provide a simple online algorithm $\gr A$ for chain partitioning online posets.
In the next two sections we show that $\gr A$ achieves the performance bound stated  in Theorem~\ref{valbound}.
If $W$ is a subset of $P$, we set $W^{x}=W \cap \{y:y\preceq x\}$ and $W^{\prec x}=W \cap \{y:y\prec x\}$.

\subsection{Overview}
\label{subsec:overview}

We define $\gr A$ using three Procedures.
Consider an online poset $P^\prec=(V,\le_P,\prec)$.
\begin{enumerate}
\item Construct an online partition $V=X_1\cupdot \dots \cupdot X_{\width(P)}$ by putting
every consecutive vertex $x$ of $(V, \prec)$ to the set $X_w$, where $w$ is the least integer
with $\width(P[X_1^{\prec x} \cup \dots \cup X^{\prec x}_w \cup \{x\}])=w$.
Pick a $w$-antichain $A'_x$ in $P[X_1^{x} \cup \ldots \cup X_w^x]$ with $x\in A'_x$.

\item For every $w \in [\width(P)]$, construct an on-line poset $R^{\ll}_w$, where $R_w = (Z, \le_R)$, together with an injection $\phi:X_w \to Z$ that satisfies the property that $R_w[\phi(X_w)]$ is a
subposet of $P[X_w]$.
Thus, a partition of $R^{\ll}_w$ into chains yields a partition of $P[X_w]$ into chains.
This more complex procedure is explained in Subsection~\ref{subsec:procedure2}.

\item For every $w \in [\width(P)]$, use First-Fit to partition $R^{\ll}_w$ into chains.

\end{enumerate}
The final chain partition consists of all chains produced by Procedure (3) for $w=1,\dots,\width(P)$.\footnote{Kierstead and Trotter \cite{KT} used Procedures (1) and (3), without (2), to prove $\val(\Forb_w(\mathbf2+\mathbf2))=3w-2$.}

In Section \ref{sec:R-is-ladder-free} we show that $R_w$ is a $(2w^2+1)$-ladder free poset of width $w$.
In Section \ref{sec:firstfit} we show that  $$\val_{FF}(L_m,w)\le w^{2.5\lg (2w)+2\lg m}.$$
Then, since a chain partition of $R^\ll_w$ yields a chain partition of $P[X_w]$ with at most the same number of chains,
Theorem~\ref{valbound} follows by

\begin{equation}\label{eqM}
\val(\mathcal P_w)\le \sum^w_{j=1}\val_{FF}(R_j^{\ll}) \le w \cdot \val_{FF}(L_{2w^2+1},w)\le w^{2.5\lg (2w)+2\lg (2w^2+1)+1}
\le w^{6.5\lg w+7}.
\end{equation}

In the remaining of the paper, we write $R^\ll$ and $R$ instead of $R^\ll_w$ and $R_w$ whenever $w$ is clear from the context.

\subsection{Procedure (2)}%
\label{subsec:procedure2} %
Fix $w \in [\width(P)]$.
Note that Procedure (1) produces a partition of $V$ into $X_1 \cupdot \ldots \cupdot X_{\width(P)}$
such that $\width(P[X_1 \cup \ldots \cup X_w]) = w$.
Let $V_w = X_1 \cup \ldots \cup X_w$ and let $\mathcal M=\mathcal M(P[V_w])$ be the set of all maximum antichains in $P[V_w]$.
First, algorithm $\gr A$ constructs a chain $\mathcal A=\{A_y:y\in X_w\}$ in $(\mathcal M, \sqsubseteq_P)$.
The antichain $A_x$ is obtained from $A'_x$ and the $\sqsubseteq_P$-chain $\mathcal A^{\prec x}=\{A_y:y\in X^{\prec x}_w\}$
when $x\in X_w$ is processed.
Put
$$A_x=(A'_x \wedge U_x)\vee  D_x \text{,}$$
where
$$\mathcal U_x=\{A \in \mathcal{A}^{\prec x}: x \in D_P(A)\} \text{ and }
U_x= \left\{
\begin{array}{rl}
\bigwedge \mathcal U_x &\text{if } \mathcal U_x \neq \emptyset \\
\emptyset & \text{otherwise,}
\end{array}
\right.
$$
and
$$ \mathcal D_x= \{A \in \mathcal{A}^{\prec x}: x \in U_P(A)\} \text{ and }
D_x= \left\{
\begin{array}{rl}
\bigvee \mathcal D_x & \text{ if } \mathcal D_x \neq \emptyset \\
\emptyset & \text{otherwise,}
\end{array}
\right.
$$
see Figure \ref{fig:A_x}.
Clearly, $x\in A_x$.
Each $A\in \mathcal A^{\prec x}$ is a $w$-antichain contained in $P[V_w^{\prec x}]$,
so some $y\in A$ is comparable to $x$.
Thus $\mathcal A^{\prec x}=\mathcal D_x \cup \mathcal U_x$.
Let $\mathcal{A}^{x} = \mathcal A^{\prec x} \cup \{A_x\}$.
As $\mathcal A^{\prec x}$ is a chain and $u <_P x <_P v$ for some $u\in D_x$ and $v\in U_x$ we note that
\begin{align}
\mbox{$\mathcal{A}^{x}$ is a $\sqsubseteq_P$-chain with consecutive elements  $ D_x, A_x,  U_x$ (unless $\mathcal D_x=\emptyset$ or $\mathcal U_x=\emptyset$).}
\end{align}
Define $p(x)$ by $A_{p(x)}=D_x$ if $\mathcal D_x\ne\emptyset$ and $s(x) $ by  $A_{s(x)}= U_x$ if $\mathcal U_x\ne\emptyset$.

\begin{figure}[tbh]
\begin{center}
\begin{tikzpicture}
\path (0,1) coordinate (k1);
\fill (k1) circle (3pt);
\path (1,2) coordinate (k2);
\fill (k2) circle (3pt);
\path (2,2) coordinate (k3);
\fill (k3) circle (3pt);
\path (3,2) coordinate (k4);
\fill (k4) circle (3pt);
\path (4,2) coordinate (k5);
\fill (k5) circle (3pt);

\path (0,2) coordinate (x1);
\fill (x1) circle (3pt);
\path (1,0) coordinate (x2);
\fill (x2) circle (3pt);
\path (2,1) coordinate (x3);
\fill (x3) circle (3pt);
\path (3,1) coordinate (x4);
\fill (x4) circle (3pt);
\path (4,1) coordinate (x5);
\fill (x5) circle (3pt)
node[right,black] {$\: x$};

\path (1,1) coordinate (j2);
\fill (j2) circle (3pt);
\path (2,0) coordinate (j3);
\fill (j3) circle (3pt);
\path (3,0) coordinate (j4);
\fill (j4) circle (3pt);
\path (4,0) coordinate (j5);
\fill (j5) circle (3pt);

\draw[blue] (-.25,.75) -- (-.25,1.25) -- (.25,1.25) -- (.75,2.25) -- (4.25,2.25) --node[right,black]{$U_x$} (4.25,1.75) -- (.75,1.75) -- (.25,.75) -- (-.25,.75);

\draw[green] (-.25,1.75) -- (-.25,2.25) -- (.25,2.25) -- (.75,.25) -- (1.25,.25) -- (1.75,1.25) -- (4.5,1.25) --node[right,black]{$A_x'$} (4.5,.75) -- (1.75,.75) -- (1.25,-.25) -- (.75,-.25) -- (.25,1.75) -- (-.25,1.75);

\draw[red] (-.4,.6) -- (-.4,1.4) -- (1.15,1.4) -- (1.75,.3) -- (4.25,.3) --node[right,black]{$D_x$} (4.25,-.3) -- (1.75,-.3) -- (1.25,.6) -- (-.4,.6);

\draw (k1) -- (x1)[thick];

\draw (x2) -- (j2) -- (k2)[thick];
\draw (j2) -- (k3)[thick];

\draw (j3) -- (k2)[thick];
\draw (j3) -- (x3) -- (k3)[thick];
\draw (j3) -- (x4)[thick];

\draw (j4) -- (x3)[thick];
\draw (j4) -- (x4) -- (k4)[thick];
\draw (j4) -- (x5)[thick];

\draw (j5) -- (x5) -- (k5)[thick];

\draw (x4) -- (k5)[thick];
\draw (x5) -- (k4)[thick];

\path (6,1) coordinate (kk1);
\fill (kk1) circle (3pt);
\path (7,2) coordinate (kk2);
\fill (kk2) circle (3pt);
\path (8,2) coordinate (kk3);
\fill (kk3) circle (3pt);
\path (9,2) coordinate (kk4);
\fill (kk4) circle (3pt);
\path (10,2) coordinate (kk5);
\fill (kk5) circle (3pt);

\path (8,1) coordinate (xx3);
\fill (xx3) circle (3pt);
\path (9,1) coordinate (xx4);
\fill (xx4) circle (3pt);
\path (10,1) coordinate (xx5);
\fill (xx5) circle (3pt)
node[right,black] {$\: x$};

\path (7,1) coordinate (jj2);
\fill (jj2) circle (3pt);
\path (8,0) coordinate (jj3);
\fill (jj3) circle (3pt);
\path (9,0) coordinate (jj4);
\fill (jj4) circle (3pt);
\path (10,0) coordinate (jj5);
\fill (jj5) circle (3pt);

\draw[blue] (5.75,.75) -- (5.75,1.25) -- (6.25,1.25) -- (6.75,2.25) -- (10.25,2.25) --node[right,black]{$U_x$} (10.25,1.75) -- (6.75,1.75) -- (6.25,.75) -- (5.75,.75);

\draw[green] (5.675,.675) -- (5.675,1.325) -- (10.5,1.325) --node[right,black]{$A_x$} (10.5,.675) -- (5.675,.675);

\draw[red] (5.6,.6) -- (5.6,1.4) -- (7.15,1.4) -- (7.75,.3) -- (10.25,.3) --node[right,black]{$D_x$} (10.25,-.3) -- (7.75,-.3) -- (7.25,.6) -- (5.6,.6);


\draw (jj2) -- (kk2)[thick];
\draw (jj2) -- (kk3)[thick];

\draw (jj3) -- (kk2)[thick];
\draw (jj3) -- (xx3) -- (kk3)[thick];
\draw (jj3) -- (xx4)[thick];

\draw (jj4) -- (xx3)[thick];
\draw (jj4) -- (xx4) -- (kk4)[thick];
\draw (jj4) -- (xx5)[thick];

\draw (jj5) -- (xx5) -- (kk5)[thick];

\draw (xx4) -- (kk5)[thick];
\draw (xx5) -- (kk4)[thick];

\end{tikzpicture}
\end{center}
\caption{Constructing $A_x$ based on $A_x'$, $D_x$, and $U_x$.}\label{fig:A_x}
\end{figure}
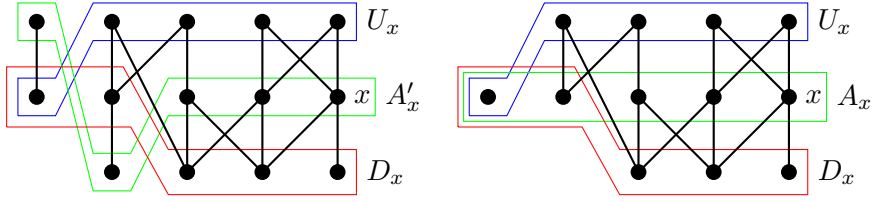

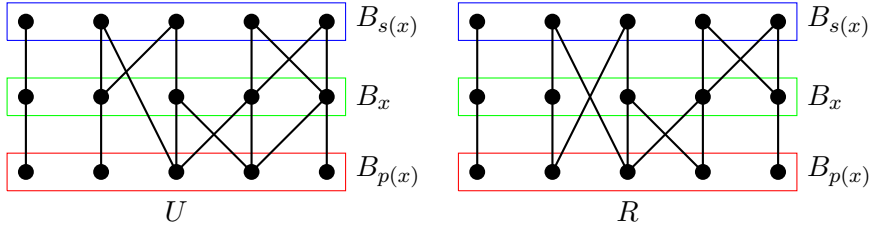
\begin{figure}[tbh]
\begin{center}
\begin{tikzpicture}
\path (0,2) coordinate (k1);
\fill (k1) circle (3pt);
\path (1,2) coordinate (k2);
\fill (k2) circle (3pt);
\path (2,2) coordinate (k3);
\fill (k3) circle (3pt);
\path (3,2) coordinate (k4);
\fill (k4) circle (3pt);
\path (4,2) coordinate (k5);
\fill (k5) circle (3pt);

\path (0,1) coordinate (x1);
\fill (x1) circle (3pt);
\path (1,1) coordinate (x2);
\fill (x2) circle (3pt);
\path (2,1) coordinate (x3);
\fill (x3) circle (3pt);
\path (3,1) coordinate (x4);
\fill (x4) circle (3pt);
\path (4,1) coordinate (x5);
\fill (x5) circle (3pt);

\path (0,0) coordinate (j1);
\fill (j1) circle (3pt);
\path (1,0) coordinate (j2);
\fill (j2) circle (3pt);
\path (2,0) coordinate (j3);
\fill (j3) circle (3pt);
\path (3,0) coordinate (j4);
\fill (j4) circle (3pt);
\path (4,0) coordinate (j5);
\fill (j5) circle (3pt);

\draw[blue] (-.25,1.75) -- (-.25,2.25) -- (4.25,2.25) --node[right,black]{$B_{s(x)}$} (4.25,1.75) -- (-.25,1.75);

\draw[green] (-.25,.75) --(-.25,1.25) -- (4.25,1.25) --node[right,black]{$B_x$}
(4.25,.75) -- (-.25,.75);

\draw[red] (-.25,-.25) -- (-.25,.25) -- (4.25,.25) --node[right,black]{$B_{p(x)}$} (4.25,-.25) --node[below,black]{$U$} (-.25,-.25);

\draw (j1) -- (x1) -- (k1)[thick];

\draw (j2) -- (x2)[thick];

\draw (j3) -- (k2)[thick];
\draw (j3) -- (x3)[thick];
\draw (j3) -- (x4)[thick];

\draw (j4) -- (x3)[thick];
\draw (j4) -- (x4)[thick];
\draw (j4) -- (x5)[thick];

\draw (j5) -- (x5)[thick];

\draw (x2) -- (k2)[thick];
\draw (x2) -- (k3)[thick];

\draw (x3) -- (k3)[thick];

\draw (x4) -- (k4)[thick];
\draw (x4) -- (k5)[thick];

\draw (x5) -- (k4)[thick];
\draw (x5) -- (k5)[thick];

\path (6,2) coordinate (k11);
\fill (k11) circle (3pt);
\path (7,2) coordinate (k21);
\fill (k21) circle (3pt);
\path (8,2) coordinate (k31);
\fill (k31) circle (3pt);
\path (9,2) coordinate (k41);
\fill (k41) circle (3pt);
\path (10,2) coordinate (k51);
\fill (k51) circle (3pt);

\path (6,1) coordinate (x11);
\fill (x11) circle (3pt);
\path (7,1) coordinate (x21);
\fill (x21) circle (3pt);
\path (8,1) coordinate (x31);
\fill (x31) circle (3pt);
\path (9,1) coordinate (x41);
\fill (x41) circle (3pt);
\path (10,1) coordinate (x51);
\fill (x51) circle (3pt);

\path (6,0) coordinate (j11);
\fill (j11) circle (3pt);
\path (7,0) coordinate (j21);
\fill (j21) circle (3pt);
\path (8,0) coordinate (j31);
\fill (j31) circle (3pt);
\path (9,0) coordinate (j41);
\fill (j41) circle (3pt);
\path (10,0) coordinate (j51);
\fill (j51) circle (3pt);

\draw[blue] (5.75,1.75) -- (5.75,2.25) -- (10.25,2.25) --node[right,black]{$B_{s(x)}$}
(10.25,1.75) -- (5.75,1.75);

\draw[green] (5.75,.75) -- (5.75,1.25) -- (10.25,1.25) --node[right,black]{$B_x$}
(10.25,.75) -- (5.75,.75);

\draw[red] (5.75,-.25) -- (5.75,.25) -- (10.25,.25) --node[right,black]{$B_{p(x)}$}
(10.25,-.25) --node[below,black]{$R$} (5.75,-.25);

\draw (j11) -- (x11) -- (k11)[thick];

\draw (j21) -- (x21)[thick];
\draw (j21) -- (k31)[thick];

\draw (j31) -- (k21)[thick];
\draw (j31) -- (x31)[thick];
\draw (j31) -- (x41)[thick];

\draw (j41) -- (x31)[thick];
\draw (j41) -- (x41)[thick];

\draw (j51) -- (x51)[thick];

\draw (x21) -- (k21)[thick];

\draw (x31) -- (k31)[thick];

\draw (x41) -- (k41)[thick];
\draw (x41) -- (k51)[thick];

\draw (x51) -- (k41)[thick];
\draw (x51) -- (k51)[thick];

\end{tikzpicture}
\end{center}
\caption{Hasse diagrams of $U[B_{p(x)} \cup B_x \cup B_{s(x)}]$ and $R[B_{p(x)} \cup B_x \cup B_{s(x)}]$.}\label{figure:R_relation}
\end{figure}

The maximum antichains $A_x$ in $\mathcal{A}$ are computed in the order in which the elements $x$ are added to the set $X_w$.
So, we may view $(\bigcup \mathcal{A}, \le_P)$ as an on-line poset with the presentation order extended by the elements from $A_x \setminus \bigcup \{ A_y: y \in X^{\prec x}_w \}$ each time a new antichain $A_x$ from $\mathcal{A}$ is computed.
It is likely that the antichains in $\mathcal A$ are not disjoint.
In the next step we slightly modify this poset by making these antichains pairwise disjoint.

When $x$ is processed, set $B_x = \{(u, A_x): u \in A_x \}$.
Let $\mathcal B=\{B_y:y\in X_w\}$, $Z=\bigcup \mathcal B$ and, following our notation,
let $\mathcal B^x=\{B_y:y\in X_w^x\}$ and $Z^x=\bigcup \mathcal B^x$.
Let $\le_U$ be the product order defined by
\begin{align}\label{defU}
(u,A) \le_{U} (u',A') &\iff u \le_P u' \textrm{ and } A \sqsubseteq A', \quad \text{for } u \in A,\ u' \in A', \text{ and } A,A' \in \mathcal{A}.
\end{align}
Define the presentation order $\ll$ of the poset $U = (Z, \le_U)$ by putting $B_x \ll B_y$ if $x \prec y$ for $x,y \in X_w$,
and by arbitrarily ordering the elements in $B_x$ for $x \in X_w$.

By \eqref{defU}, $B_x$ is an $w$-antichain in $U$.
Indeed,  if $S$ is a $(w+1)$-subset in $Z$ then there are distinct $y,z\in S$ with comparable (possibly identical) first coordinates.
By \eqref{defU}, they are comparable in $U$.
Thus $\width( U)=w$, and $\mathcal B$  partitions $Z$  into disjoint maximum antichains.
Moreover, $\mathcal B$ is a $\sqsubseteq_U$-chain.
Note that the antichain $B_x$ and the relation $\le_U$ between the elements in $B_x$ and the elements in the set $Z^{\prec x}$
can be computed when $x$ is processed.


In its last step, Procedure (2) constructs an online poset $R^{\ll}=(Z, \le_R, \ll)$,
where $R=(Z,\le_R)$ is obtained from $(Z,\le_U)$ by deleting some non-Dilworth edges of $(Z,\le_U)$.
Suppose $\le_R$ restricted to $Z^{\prec x}$ is already computed.
When $x$ is processed, the edges $(v,u)$ of $U[Z^x]$ with $u\in B_x$ are deleted unless there is a Dilworth edge
$(u',u)\in U[B_{p(x)}\cup B_x]$ such that $v \le_R u'$ (possibly $v=u'$) in $(Z^{\prec x}, \le_R)$.
Dually, the edges $(u,v)$ of $U[Z^x]$ with $u\in B_x$ are deleted unless there is a Dilworth edge
$(u,u')\in U[B_{x}\cup B_{s(x)}]$ such that $u'\le_R v$ (possibly $v=u'$) in $(Z^{\prec x}, \le_R)$.
This completes the definition of $R$.
Figure \ref{figure:R_relation} illustrates the derivation of  $\le_R$ from $\le_U$.

Note that $R[Z^x]$ can be computed from $P^x$.
Note that if $(v,u)$ is a Dilworth edge in $U[Z^x]$ then $v <_R u$ in $R$.
We prove this by induction on $\prec$.
The claim holds for the set $Z^{y}$, where $y$ is the first vertex in $(V, \prec)$ added to $X_w$.
Suppose the claim holds for $Z^{\prec x}$.
Assume that $v\ll u$ and $u \in B_x$ (the other cases are handled similarly).
Then there is a Dilworth partition $\mathcal{C}$ with a chain $C$ such that $v,u\in C$.
Thus there is $z\in C\cap B_{p(x)}$.
Since $\mathcal C$ restricted to $Z^{\prec x}$ is a chain partition of $U[Z^{\prec x}]$,
we get $v \le_R z$ by inductive hypothesis.
As $v \le_R z$ and $(z,u)$ is Dilworth in $U(B_{p(x)} \cup B_x)$, $(v,u)$ is not deleted, and hence $v \le_R u$.

Let $\phi: X_w \to Z$ be a mapping defined $\phi(x) = (x,A_x)$.
Let $x,x' \in X_w$.
Clearly, $\phi(x) \le_R \phi(x')$ is equivalent to $(x,A_x) \le_R (x',A_{x'})$, which
yields $x \le_P x'$.
Hence, a chain partition of $R^{\ll}_w$ induces a chain partition of $P[X_w]$ into at most the same number of chains:
indeed, it is enough to assign $x \in X_w$ to a chain labeled $i$ if $\phi(x)=(x,A_x)$ is assigned to a chain $i$.

\begin{lemma}
The relation $R=(Z,\le_R)$ is a width $w$ poset, and for all $x,y\in X_w$ with $B_x\sqsubset_RB_y$:
\begin{enumerate}[label=(R\theenumi), series=Rconditions]
  \item \label{R:core} $R[B_{p(x)} \cup B_x]$   and $R[B_x \cup B_{s(x)}]$ are  cores;
  \item \label{R:transitivity} Suppose  $u\in B_x$, $v\in B_y$ and $u <_R v$. If $x\prec y$
  then  there is
   $v' \in B_{p(y)}$ with $u \le_R v' <_R v$; else
  there is $u' \in B_{s(x)}$ with $u <_{R} u' \le_R v$.
    \item $\mathcal B$ is a partition of $Z$ into maximum antichains. \label{R:sum}
  \item \label{R:linear} $\mathcal B$ is a chain in $\sqsubseteq_R$.
  \item $R[B_x,B_y]$ is a core.\label{core+}
  \item Let $z\in X_w$ be such that $B_x \sqsubseteq_R B_z \sqsubseteq_R B_y$
  and suppose that for every $z' \in X_w$ such that $B_z \sqsubseteq_R B_{z'} \sqsubseteq_R B_y$
  ($B_x \sqsubseteq_R B_{z'} \sqsubseteq_R B_z$) we have $z \preceq z'$.
  Then, for all $u\in B_x$ and $v\in B_y$ with $u \le_R v$, there is $z'' \in B_z$ with $u \le_R z'' \le_R v$.\label{transitivity+}

\end{enumerate}
\end{lemma}

\begin{proof}
Conditions \ref{R:core} and \ref{R:transitivity} follow immediately from the definition of $R$.

First we prove that $R$ is a poset of width $w$.
As $R$ is obtained from the poset $U$ by removing some non-loops, $R$ is reflexive and antisymmetric.
For transitivity, argue by induction on $\prec$.
Suppose $u<_Rv<_Rw$.
Then there are distinct $x,y,z\in X^w$ with $u\in B_x$, $v\in B_y$, $w\in B_z$, and $B_x\sqsubset_R B_y \sqsubset_R B_z$.
Let $s={\prec}\mbox{-}\!\max\{x,y,z\}$.
If $s=y$ then using \ref{R:transitivity} there are $v'\in B_{p(y)}$ and $v''\in B_{s(y)}$ such that $u\le_Rv'\le_Rv$,  $v\le_Rv''\le_Rw$, $(v',v)$ is Dilworth in $U[B_{p(y)} \cup B_{y}]$ and $(v,v'')$ is Dilworth in $U[B_{y} \cup B_{s(y)}]$;
thus $(v',v'')$ is Dilworth in $U[B_{p(y)} \cup B_{s(y)}]$, $v' <_R v''$, and $u <_R w$ by induction.
The other two cases are similar, but easier.
Thus $R$ is a poset.
As no Dilworth edges are removed from $U$ to form $R$,
$\width(R)=\width(U)=w$.
Thus \ref{R:sum} and \ref{R:linear} also hold.

We prove \ref{core+} by induction on $\prec$.
Assume $x \prec y$.
The case $y\prec x$ is dual.
By \ref{R:core}, $R[B_{p(y)} \cup B_y]$ is a core.
If $x=p(y)$ we are done.
Otherwise, $R[B_x \cup B_{p(y)}]$ is a core by induction.
Thus $R[B_x,B_y]$ is a core by definition of $\le_R$.
So, \ref{core+} holds.

We prove \ref{transitivity+} by induction on $\prec$.
Suppose $u\in B_x$ and $v\in B_y$ with $u\le_R v$.
The claim holds if $z = x$ or $z = y$.
Suppose $z \in X_w$ is such that $B_x \sqsubset_R B_z \sqsubset_R B_y$ and
$z \preceq z'$ for any $B_z \sqsubseteq_R B_{z'} \sqsubseteq_R B_y$.
Suppose $x \prec y$.
By \ref{R:transitivity}, there is $w'\in B_{p(y)}$ with $u\leqslant_Rw'<_Rv$.
If $z=p(y)$ we are done.
Otherwise, as $z\prec p(y)$, there is $w\in B_z$ with $u <_R w <_R w'<_R y$ by induction, and hence \ref{transitivity+} holds.
Suppose $y \prec x$.
By \ref{R:transitivity}, there is $w'\in B_{s(x)}$ with $u \leqslant_R w' <_R v$.
If $z=s(x)$ we are done.
Otherwise, as $B_{s(x)} \sqsubset_{R} B_z \sqsubseteq_R B_y$,
there is $w\in B_z$ with $u <_R w' <_R w <_R y$ by induction.
Thus \ref{transitivity+} holds.
\end{proof}

In \cite{BK-FOCS,BK15} a width $w$ poset $R$ is defined to be \emph{regular} if it has a partition $\mathcal B$ satisfying \ref{R:core}--\ref{R:linear}.
An example of a regular poset is given in Figure \ref{figure:regexample}.

\begin{figure}[tbh]
\begin{center}
\begin{tikzpicture} 

\path (0,0) coordinate (R0); 

\foreach \i in {1,...,4}{
\path (R0) ++(\i,5) coordinate (f\i);
\fill (f\i) circle (3pt);
\path (R0) ++(\i,4) coordinate (e\i);
\fill (e\i) circle (3pt);
\path (R0) ++(\i,2) coordinate (c\i);
\fill (c\i) circle (3pt);
\path (R0) ++(\i,0) coordinate (a\i);
\fill (a\i) circle (3pt);
}

\foreach \i in {1,...,4}{ \draw (a\i) -- (f\i);}

\draw (e4) -- (f3);
\draw (e4) -- (f2);
\draw (e4) -- (f1);
\draw (e3) -- (f4);
\draw (e2) -- (f3);
\draw (e1) -- (f2);

\draw (c4) -- (e3);
\draw (c3) -- (e4);
\draw (c3) -- (e2);
\draw (c2) -- (e1);
\draw (c1) -- (e3);
\draw (c1) -- (e2);

\draw (a4) -- (c3);
\draw (a3) -- (c4);
\draw (a2) -- (c1);
\draw (a1) -- (c2);

\draw (a3) -- (c2);
\draw (a2) -- (c3);

\draw (f1) ++ (-1/4,-1/4) rectangle ++ (3.5,1/2);
\path (f1) ++ (-.6,0)node{$A_{3}$};
\path (f4) ++ (.6,0)node{$\phantom{A_{3}}$};
\draw (e1) ++ (-1/4,-1/4) rectangle ++ (3.5,1/2);
\path (e1) ++ (-.6,0)node{$A_{1}$};
\draw (c1) ++ (-1/4,-1/4) rectangle ++ (3.5,1/2);
\path (c1) ++ (-.6,0)node{$A_{2}$};
\draw (a1) ++ (-1/4,-1/4) rectangle ++ (3.5,1/2);
\path (a1) ++ (-.6,0)node{$A_{4}$};

\path (5,0) coordinate (R0); 

\foreach \i in {1,...,4}{
\path (R0) ++(\i,5) coordinate (f\i);
\fill (f\i) circle (3pt);
\path (R0) ++(\i,4) coordinate (e\i);
\fill (e\i) circle (3pt);
\path (R0) ++(\i,3) coordinate (d\i);
\fill (d\i) circle (3pt);
\path (R0) ++(\i,2) coordinate (c\i);
\fill (c\i) circle (3pt);
\path (R0) ++(\i,0) coordinate (a\i);
\fill (a\i) circle (3pt);
}

\foreach \i in {1,...,4}{ \draw (a\i) -- (f\i);}

\draw (e4) -- (f3); \draw (e4) -- (f2); \draw (e4) -- (f1); \draw (e3) -- (f4); \draw (e2) -- (f3); \draw (e1) -- (f2);

\draw (c4) -- (e3); \draw (c3) -- (e4); \draw (c3) -- (e2); \draw (c2) -- (e1); \draw (c1) .. controls +(20:1) and +(260:1) .. (e3); \draw (c1) -- (e2);

\draw (a4) -- (c3); \draw (a3) -- (c4); \draw (a2) -- (c1); \draw (a1) -- (c2);

\draw (a3) -- (c2); \draw (a2) -- (c3);

\draw (f1) ++ (-1/4,-1/4) rectangle ++ (3.5,1/2);
\path (f1) ++ (-.6,0)node{$A_{3}$};
\path (f4) ++ (.6,0)node{$\phantom{A_{3}}$};
\draw (e1) ++ (-1/4,-1/4) rectangle ++ (3.5,1/2);
\path (e1) ++ (-.6,0)node{$A_{1}$};
\draw (d1) ++ (-1/4,-1/4) rectangle ++ (3.5,1/2);
\path (d1) ++ (-.6,0)node{$A_{5}$};
\draw (c1) ++ (-1/4,-1/4) rectangle ++ (3.5,1/2);
\path (c1) ++ (-.6,0)node{$A_{2}$};
\draw (a1) ++ (-1/4,-1/4) rectangle ++ (3.5,1/2);
\path (a1) ++ (-.6,0)node{$A_{4}$};

\path (10,0) coordinate (R0); 

\foreach \i in {1,...,4}{
\path (R0) ++(\i,5) coordinate (f\i);
\fill (f\i) circle (3pt);
\path (R0) ++(\i,4) coordinate (e\i);
\fill (e\i) circle (3pt);
\path (R0) ++(\i,3) coordinate (d\i);
\fill (d\i) circle (3pt);
\path (R0) ++(\i,2) coordinate (c\i);
\fill (c\i) circle (3pt);
\path (R0) ++(\i,1) coordinate (b\i);
\fill (b\i) circle (3pt);
\path (R0) ++(\i,0) coordinate (a\i);
\fill (a\i) circle (3pt);
}

\foreach \i in {1,...,4}{ \draw (a\i) -- (f\i);}

\draw (e4) -- (f3); \draw (e4) -- (f2); \draw (e4) -- (f1); \draw (e3) -- (f4); \draw (e2) -- (f3); \draw (e1) -- (f2);

\draw (c4) -- (e3); \draw (c3) -- (e4); \draw (c3) -- (e2); \draw (c2) -- (e1); \draw (c1) .. controls +(20:1) and +(260:1) .. (e3); \draw (c1) -- (e2);

\draw (b4) -- (c3); \draw (b3) -- (c4); \draw (b2) -- (c1); \draw (b1) -- (c2);

\draw (a3) -- (c2); \draw (a2) -- (c3);

\draw (f1) ++ (-1/4,-1/4) rectangle ++ (3.5,1/2); \path (f1) ++ (-.6,0)node{$A_{3}$}; \path (f4) ++ (.6,0)node{$\phantom{A_{3}}$}; \draw (e1) ++ (-1/4,-1/4) rectangle ++ (3.5,1/2); \path (e1) ++ (-.6,0)node{$A_{1}$}; \draw (d1) ++ (-1/4,-1/4) rectangle ++ (3.5,1/2); \path (d1) ++ (-.6,0)node{$A_{5}$}; \draw (c1) ++ (-1/4,-1/4) rectangle ++ (3.5,1/2); \path (c1) ++ (-.6,0)node{$A_{2}$}; \draw (b1) ++ (-1/4,-1/4) rectangle ++ (3.5,1/2); \path (b1) ++ (-.6,0)node{$A_{6}$}; \draw (a1) ++ (-1/4,-1/4) rectangle ++ (3.5,1/2); \path (a1) ++ (-.6,0)node{$A_{4}$};

\end{tikzpicture}
\end{center}
\caption{Regular poset $P^{\ll}$ with $P=(\bigcup_{i=1}^{6}A_i, \leq_P)$
and the presentation order given by $A_1 \ll A_2 \ll \ldots \ll A_6$ (the order $\ll$ inside each $A_i$ is arbitrary).}
\label{figure:regexample}
\end{figure}
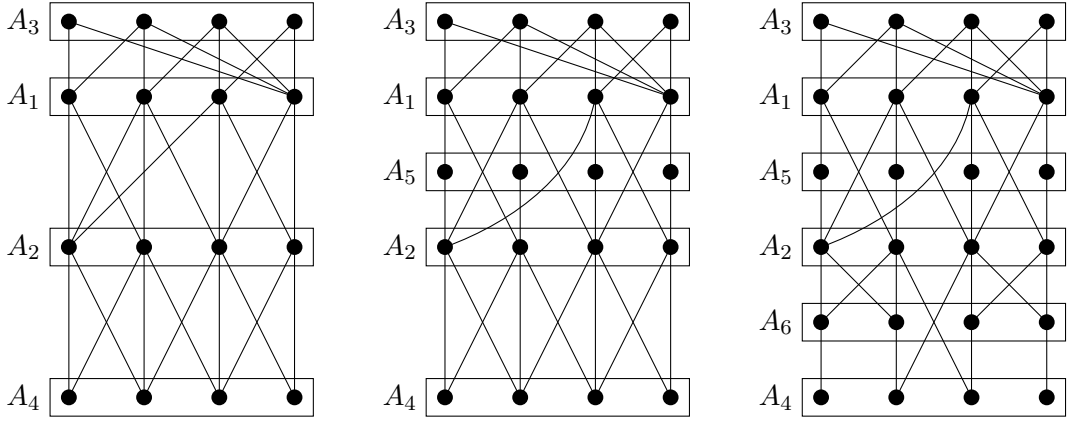

\section{$R$ is $L_{2w^2+1}$-ladder-free}
\label{sec:R-is-ladder-free}

In this section we prove that $R \in \Forb_w(L_{2w^2+1})$, which yields Lemma \ref{ladobs} as a corollary.
For any $u \in R$, let $B(u)$ be the antichain with $u\in B(u)\in \mathcal B$.
Consider an arbitrary $m$-ladder  $L=L(x_1\dots x_m;y_1\dots y_m)$ in $R$.
Call $L$ \emph{canonical} if $B(y_i)\sqsubseteq_R B(x_{i+1})$ for all $i\in[m-1]$.

\begin{prop} \label{claim:canonical-ladder}
If $L=L(x_1\dots x_m;y_1\dots y_m)\subseteq R$ is canonical then $m\le w$.
\end{prop}
\begin{proof}
See Figure~\ref{Ugrow}.
As $\width (R)\le w$, it suffices to show by induction on $i$ that
\begin{equation} \label{eq:canonical-ladder}
\abs{U_R [y_1] \cap B(y_i)} \ge i \text{ for } i \in [m].
\end{equation}

The base step $i=1$ holds, since $y_1\in U_R [y_1] \cap B(y_1) $, so assume $1 < i \le m$.
As $L$ is  canonical, $B(y_{i-1}) \sqsubseteq_R B(x_{i})$.
Thus there is $z\in B(y_{i-1})$ such that $z\le_Rx_i\le_Ry_i$.
Since $y_1 \parallel_R x_i$, we have $y_1 \parallel_R z$.
Thus $z\notin S := U_R[y_1] \cap B(y_{i-1})$.
By  induction, $\abs{S} \ge i-1$.
By~\ref{core+}, $R[B(y_{i-1}) \cup B(y_i)]$ is a core with Dilworth edge $zy_i$.
Let $\mathcal C$ be a Dilworth partition  of $R[B(y_{i-1}) \cup B(y_i)]$ with $z$
and $y_{i}$ in the same chain.
Each vertex of $S$ is matched in $\mathcal{C}$ to a distinct
vertex of $B(y_i)$, different than $y_i$ (see Figure~\ref{Ugrow}) as $z\notin S$.
Consequently, $\abs{U_R[y_1] \cap B(y_i)} \ge \abs{S} +1 \ge i-1+1 =i $.
This proves \eqref{eq:canonical-ladder}.
\end{proof}

\begin{figure}[tbh]
\begin{center}
\begin{tikzpicture}
\path (2,0) coordinate (y1);
\fill (y1) circle (3pt)
node[left]{$y_1$};

\path (1,1.5) coordinate (z);
\fill (z) circle (3pt)
node[left]{$z$};
\path (2,1.5) coordinate (yi1);
\fill (yi1) circle (3pt)
node[right]{$y_{i-1}$};
\path (3,1.5) coordinate (u1);
\fill (u1) circle (3pt);
\path (4,1.5) coordinate (u2);
\fill (u2) circle (3pt);
\path (5,1.5) coordinate (u3);
\fill (u3) circle (3pt);

\path (1,2.5) coordinate (zi);
\fill (zi) circle (3pt)
node[left]{$x_i$};

\path (2,3.5) coordinate (yi);
\fill (yi) circle (3pt)
node[right]{$y_{i}$};
\path (3,3.5) coordinate (v1);
\fill (v1) circle (3pt);
\path (4,3.5) coordinate (v2);
\fill (v2) circle (3pt);
\path (5,3.5) coordinate (v3);
\fill (v3) circle (3pt);
\path (6,3.5) coordinate (v4);
\fill (v4) circle (3pt);

\draw (-.4,-.35) -- (-.4,.35) -- (6.4,.35) --node[right,black]{$B(y_1)$} (6.4,-.35)
-- (-.4,-.35);

\draw (-.4,1.15) -- (-.4,1.85) -- (6.4,1.85) --node[right,black]{$B(y_{i-1})$}
(6.4,1.15) -- (-.4,1.15);

\draw (1.8,1.3) -- (1.8,1.7) -- (5.2,1.7) --node[right,black]{$S$} (5.2,1.3) --
(1.8,1.3);

\draw (-.4,2.15) -- (-.4,2.85) -- (6.4,2.85) --node[right,black]{$B(x_i)$}
(6.4,2.15) -- (-.4,2.15);

\draw (-.4,3.15) -- (-.4,3.85) -- (6.4,3.85) --node[right,black]{$B(y_i)$}
(6.4,3.15) -- (-.4,3.15);

\draw (2.8,3.3) -- (2.8,3.7) -- (6.2,3.7) -- (6.2,3.3) -- (2.8,3.3);

\draw (y1) -- (yi1) -- (yi)[thick];
\draw (yi1) -- (v1)[thick];
\draw (y1) -- (u1) -- (v2)[thick];
\draw (y1) -- (u2) -- (v3)[thick];
\draw (y1) -- (u3) -- (v4)[thick];
\draw (z) -- (zi) -- (yi)[thick];
\end{tikzpicture}
\end{center}
\caption{The intersections of $U_P [y_1]$ with $B(y_{i-1})$ and
$B(y_i)$.}\label{Ugrow}
\end{figure}
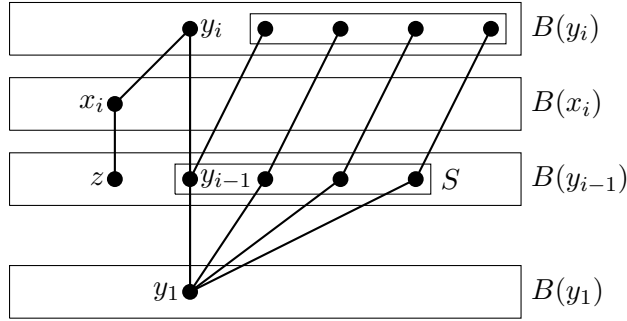

\begin{prop} \label{claim:non-canonical-ladder}
If $L=L(x_1\dots x_m;y_1\dots y_m)\subseteq R$ with $m\ge2w+1$ then $B(y_1) \sqsubseteq_RB(x_{2w+1})$.
\end{prop}
\begin{proof}
See Figure~\ref{khits}.
Assume to the contrary that $B(x_{2w+1}) \sqsubset_RB(y_1)$.
It follows that
$$B(x_1) \sqsubset_R \ldots \sqsubset_R B(x_{2w+1}) \sqsubset_R B(y_1) \sqsubset_R  \ldots \sqsubset_R B(y_{2w+1}).$$
Let $z \in X_w$ be the $\prec$-least index with $B(x_{w+1})
\sqsubseteq_R B_z \sqsubseteq_R B(y_{w+1})$.
If $B_z \sqsubset_RB(y_1)$ then set $I=[w+1]$; else
set $I=\{w+1,\dots,2w+1\}$. Regardless, $|I|=w+1$, and by \ref{transitivity+}, there are $z_i$ with $x_i<_Rz_i<_Ry_i$ for all $i\in I$
(see Figure~\ref{khits}).
As $|B_z|=w$ there are $i,j\in I$ with $i<j$ and $z_i=z_j$. Then $x_j<_Rz_j=z_i<_Ry_i$, a contradiction with $x_j \parallel_R y_i$.
\end{proof}

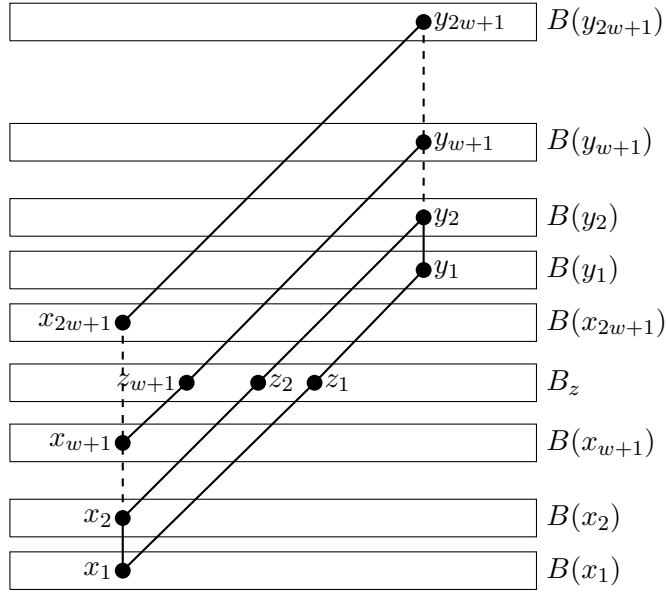
\begin{figure}[tbh]
\begin{center}
\begin{tikzpicture}

\path (5.5,7.55) coordinate (y2w);
\fill (y2w) circle (3pt)
node[right]{$y_{2w+1}$};

\path (5.5,5.95) coordinate (yw);
\fill (yw) circle (3pt)
node[right]{$y_{w+1}$};

\path (5.5,4.95) coordinate (y2);
\fill (y2) circle (3pt)
node[right]{$y_2$};

\path (5.5,4.25) coordinate (y1);
\fill (y1) circle (3pt)
node[right]{$y_1$};

\draw (0,7.3) -- (0, 7.8) -- (7, 7.8) --node[right,black]{$B(y_{2w+1})$} (7,7.3) -- (0,7.3);

\draw (0,5.7) -- (0, 6.2) -- (7, 6.2) --node[right,black]{$B(y_{w+1})$} (7,5.7) -- (0,5.7);

\draw (0,4.7) -- (0, 5.2) -- (7, 5.2) --node[right,black]{$B(y_2)$} (7,4.7) -- (0,4.7);

\draw (0,4) -- (0, 4.5) -- (7, 4.5) --node[right,black]{$B(y_1)$} (7,4) -- (0,4);

\path (1.5,3.55) coordinate (x2w);
\fill (x2w) circle (3pt)
node[left]{$x_{2w+1}$};

\path (1.5,1.95) coordinate (xw);
\fill (xw) circle (3pt)
node[left]{$x_{w+1}$};

\path (1.5,0.95) coordinate (x2);
\fill (x2) circle (3pt)
node[left]{$x_2$};

\path (1.5,0.25) coordinate (x1);
\fill (x1) circle (3pt)
node[left]{$x_1$};

\path (4.05,2.75) coordinate (z1);
\fill (z1) circle (3pt)
node[right]{$z_1$};

\path (3.3,2.75) coordinate (z2);
\fill (z2) circle (3pt)
node[right]{$z_2$};

\path (2.35,2.75) coordinate (zw);
\fill (zw) circle (3pt)
node[left]{$z_{w+1}$};

\draw (0,3.3) -- (0, 3.8) -- (7, 3.8) --node[right,black]{$B(x_{2w+1})$} (7,3.3) -- (0,3.3);

\draw (0,2.5) -- (0, 3.0) -- (7, 3.0) --node[right,black]{$B_z$} (7,2.5) -- (0,2.5);

\draw (0,1.7) -- (0, 2.2) -- (7, 2.2) --node[right,black]{$B(x_{w+1})$} (7,1.7) -- (0,1.7);

\draw (0,0.7) -- (0, 1.2) -- (7, 1.2) --node[right,black]{$B(x_2)$} (7,0.7) -- (0,0.7);

\draw (0,0) -- (0, 0.5) -- (7, 0.5) --node[right,black]{$B(x_1)$} (7,0) -- (0,0);

\draw (x1) -- (z1)[thick];
\draw (z1) -- (y1)[thick];

\draw (x2) -- (z2)[thick];
\draw (z2) -- (y2)[thick];

\draw (xw) -- (zw)[thick];
\draw (zw) -- (yw)[thick];

\draw (x2w) -- (y2w)[thick];

\draw (x1) -- (x2)[thick];

\draw (x2) -- (x2w)[thick, dashed];

\draw (y1) -- (y2)[thick];

\draw (y2) -- (y2w)[thick, dashed];

\end{tikzpicture}
\caption{The ladder $L$ and the antichain $B_z$.}\label{khits}
\end{center}
\end{figure}

\begin{lemma}\label{eladder}
$R \in Forb(L_{2w^2+1})$. 
\end{lemma}
\begin{proof}
Suppose $L=L(x_1\dots x_{2w^2+1};y_1\dots y_{2w^2+1})\subseteq R$.
By Proposition~\ref{claim:non-canonical-ladder}, for any $i,j \in [2w^2+1]$ with $j-i \ge 2w$ we must have $B(y_i) \sqsubseteq_R B(x_j)$.
Thus, the subposet induced by the vertices $$\bigcup_{0 \le i \le w} \{ x_{2wi+1} , y_{2wi+1} \}$$ is a canonical ladder with $w+1$ rungs,
which contradicts Proposition~\ref{claim:canonical-ladder}.
\end{proof}

\section{First-Fit on ladder-free posets}
\label{sec:firstfit}
In this section we prove $\val_{\FF} (L_m, w) \le  w^{2.5 \lg (2w) + 2 \lg m}$ for $m, w \in \mathbb Z^+$, which shows Lemma~\ref{ladupbound}.
This proof was already published in \cite{KSm}, here we present its shortened version to keep the paper self-contained. 
Consider a Grundy coloring $\gr{g}$ of a poset $P$.
Let $C = \{x_1,\ldots,x_k\}$ be a chain of $P$ such that $\gr g(x_1) < \dots < \gr g(x_k)$;
we call $C$ \emph{ascending} if $x_1<_P x_2 <_P \ldots <_P x_k$ (see Figure~\ref{figure:ascending_ladder}) and we call $C$ \emph{descending} if $x_1 >_P x_2 >_P \ldots >_P x_k$.

The next two propositions are the combinatorial tools for the upcoming arguments in the proof of Lemma~\ref{ladupbound}.
The first one is just a restatement of the Erd\H{o}s-Szekeres Theorem and the second one presents conditions for ascending and descending chains in a poset with forbidden ladder.

\begin{prop}
\label{claim:erdosz-szekeres}
Consider a poset $P$ and its Grundy coloring $\gr{g}$.
Let $C$ be a chain in $P$ such that all $\gr g(c)$ for $c\in C$ are distinct.
If the length of every ascending subchain of $C$ is at most $s$ and the length of every descending subchain $C$ is at most $t$, then $|C| \le  st$.
\end{prop}
\begin{prop}\label{claim:ascending_ladders}
Suppose $P\in \Forb(L_m,w)$ and $w\geq2$.
Let $x_1<\ldots<x_k$ be an ascending (resp. let $x_1 > \ldots > x_k$ be a descending) chain in $P$ and for each $i \in [k-1]$ let
$y_i$ be a $\gr{g}(x_i)$-witness for $x_{i+1}$.
Then for all $i,j$ with $1\le i<j\le k$,
\begin{enumerate}[label=(C\theenumi)]
 \item\label{cond:1} $x_i<_P y_i$ (resp. $x_i >_P y_i$);
 \item\label{cond:2} $y_i \not >_P x_j$ (resp. $y_i \not <_P x_j$); and
 \item\label{cond:3} if $y_h \parallel_P x_k$ for all $h\in [k-1]$, then $k \le m(w-1)$.
\end{enumerate}
\end{prop}
\begin{proof}
We consider the case that $x_1<\ldots<x_k$ is an ascending chain; the other case can be proved analogically.
Observe that
\ref{cond:1}
and \ref{cond:2} follow immediately from definitions (see Figure \ref{figure:ascending_ladder}).
To show \ref{cond:3} assume to the contrary that $k > m(w-1)$. Then $k \geq (m-1)(w-1)+2$ as $w \geq 2$.
The subposet  $P_0:=P[\{y_1,\ldots,y_{k-1}\}]$ has width at most $w-1$ as
$y_h\parallel_P x_k$ for all $h \in [k-1]$ by  hypothesis. Since $|P_0|\ge(m-1)(w-1)+1$, there is a chain $y_{i_1}<_P \ldots<_Py_{i_m}$ in $P_0$,
by Dilworth's Theorem.
By \ref{cond:2}, ${i_1} < \ldots < {i_m}$. Thus by \ref{cond:1} and hypothesis, $P[\{x_{i_1},y_{i_1}, \ldots, x_{i_m}, y_{i_m}\}]$
is an $m$-ladder, contradicting $P \in \Forb(L_m,w)$.
\end{proof}

\begin{figure}[tbh]
\begin{center}
\begin{tikzpicture}

\draw (0,-.25) node[below] {$P_{\gr g(x_1)}$} -- (0,4.5);
\draw (2,-.25) node[below] {$P_{\gr g(x_2)}$} -- (2,4.5);
\draw (4,-.25) node[below] {$P_{\gr g(x_3)}$} -- (4,4.5);
\draw (6,-.25) node[below] {$P_{\gr g(x_4)}$} -- (6,4.5);
\draw (8,-.25) node[below] {$P_{\gr g(x_5)}$} -- (8,4.5);
\draw (10,-.25) node[below] {$P_{\gr g(x_6)}$} -- (10,4.5);

\draw (0,0) -- (10, 4);

\path (0,0) coordinate (x1);
\path (2,.8) coordinate (x2);
\path (4,1.6) coordinate (x3);
\path (6,2.4) coordinate (x4);
\path (8,3.2) coordinate (x5);
\path (10,4) coordinate (x6);

\fill (x1) circle (3pt) node[below=3,right=1,black] {$x_1$};
\fill (x2) circle (3pt) node[below=3,right=1,black] {$x_2$};
\fill (x3) circle (3pt) node[below=3,right=1,black] {$x_3$};
\fill (x4) circle (3pt) node[below=3,right=1,black] {$x_4$};
\fill (x5) circle (3pt) node[below=3,right=1,black] {$x_5$};
\fill (x6) circle (3pt) node[below=3,right=1,black] {$x_6$};

\path (0,1) coordinate (y1);
\path (2,1.8) coordinate (y2);
\path (4,2.6) coordinate (y3);
\path (6,3.4) coordinate (y4);
\path (8,4.2) coordinate (y5);

\fill (y1) circle (3pt) node[left,black] {$y_1$};
\fill (y2) circle (3pt) node[left,black] {$y_2$};
\fill (y3) circle (3pt) node[left,black] {$y_3$};
\fill (y4) circle (3pt) node[left,black] {$y_4$};
\fill (y5) circle (3pt) node[left,black] {$y_5$};

\draw (x2) -- (y1)[dashed];
\draw (x3) -- (y2)[dashed];
\draw (x4) -- (y3)[dashed];
\draw (x5) -- (y4)[dashed];
\draw (x6) -- (y5)[dashed];

\end{tikzpicture}
\end{center}
\caption{The ascending chain $x_1<_P\ldots<_Px_6$. The point $y_i$ is a $\gr{g}(x_i)$-witness of $x_{i+1}$.}
\label{figure:ascending_ladder}
\end{figure}
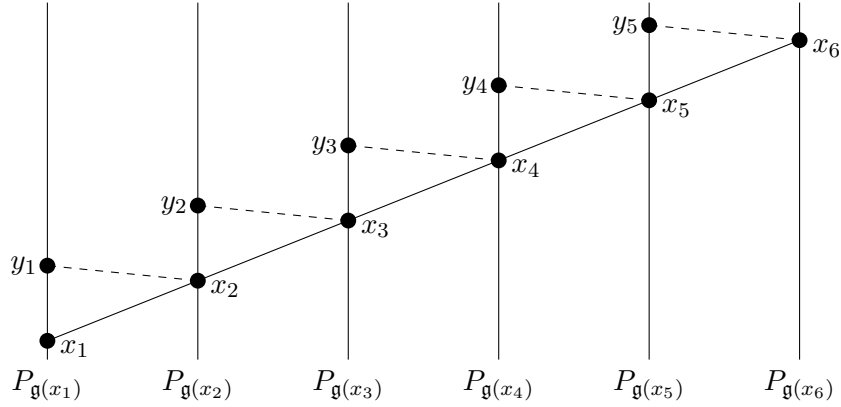

\begin{proof}[Proof of Lemma~\ref{ladupbound}]
We  argue by induction on $w=\width(P)$.
The base step $w=1$ is trivial. 
Now fix $w$, and assume the lemma holds for all smaller values of $w$.

Let $P=(V,\le_P)$ be a poset of width $w$ such that $P \in \Forb (L_m)$, let $\gr{g}: V \to [n]$ be an $n$-Grundy coloring of $P$
with $n = \val_{\FF}(L_m, w)$, and let $C_i = g^{-1}(i)$.
We must show that $n \le w^{2.5 \lg (2w) + 2 \lg m}.$

Pick a maximum antichain  $A\in\mathcal V_P$ with $N:=\min_{a\in A}\gr g(a)$ maximum, i.e.,
$$ N=\min_{a\in A}\gr g(a)= \max_{B \in \mathcal V_P} \min_{b \in B} \gr{g}(b).$$
Then $H:=P[C_{N+1} \cup \dots \cup C_n]$ has width at most
$w-1$. As $\gr g-N$ is a Grundy coloring of $H$,
\begin{equation}
n \le N + \val_{\FF} (L_m, w-1). \label{nbound}
\end{equation}

Let $D=D(A)$ and $U=U(A)$.
As $A$ is a maximum antichain, $P=D\cupdot A\cupdot U$.
Call a vertex $x$ \emph{special} if $|I_P(x)\cap A|\ge w/2$.
For $i\in [N-1]$, set $q_i^-=\max ( C_i\cap D )$ and $q_i^+=\min ( C_i \cap U )$.
Hence $q_i^{-}$ and $q_i^{+}$ are consecutive on $C_i$.
Each $a\in A$ has an $i$-witness, and so satisfies $q_i^-\parallel_P a$ or $q_i^+\parallel_P a$.
Thus, $q_i^-$ or $q_i^+$ is special.
Pick a special vertex $q_i\in\{q_i^-,q_i^+\}$ and pick $r_i\in C_i$ so that $r_i$ is a minimal special vertex in $C_i$ if $q_i \in D$, and $r_i$ is a maximal special vertex in $C_i$ if $q_i \in U$ (it might happen that $r_i=q_i$).
Call $q_i$ the \emph{near witness} and $r_i$ the \emph{far witness}.
Set $R = \{ r_1, \ldots , r_{N-1} \}$.
The next claim completes our recursion for $\val_{FF}(L_m,w)$.
\begin{claim} \label{maxint}
$|S| \le \frac{1}{2} m^2w (w-1)^2 \val_{\FF} (L_m, \lfloor w/2 \rfloor )$ for all chains $S$ with $S\subseteq R\cap D$ or $S\subseteq R\cap U$.
\end{claim}
\begin{proof}Let $S\subseteq R\cap U$; the case $S\subseteq R\cap D$ is dual.
Recall that all $\gr g(r_i)$ are distinct.
It gives that also all $\gr g(s)$, for $s\in S$, are distinct.
Therefore, by Proposition~\ref{claim:erdosz-szekeres}, it suffices to show:
\begin{enumerate}[label=(T\theenumi)]
 \item\label{asc:T} the size of any ascending chain in $S$ is at most $m(w-1)$; and
 \item\label{desc:T} the size of any descending chain in $S$ is at most $\frac{w}{2} m (w-1)\val_{\FF} (L_m, \lfloor w/2 \rfloor )$.
\end{enumerate}

For \ref{asc:T}, let $x_1<_P\ldots<_Px_k$ be any ascending chain in $S$.
For each $i\in [k-1]$, pick a $\gr{g}(x_i)$-witness $y_i$ for $x_{i+1}$.
Using Proposition~\ref{claim:ascending_ladders}, \ref{cond:2} implies
$y_i \not>_P x_k$.
Suppose $y_i <_P x_k$.
By \ref{cond:1}, $x_i<_Py_i$.
Then $y_i \in U$ and $y_i$ is special as $x_k$ is special.
As $\gr{g}(x_i)=\gr g(y_i)$ this contradicts the choice of $x_i$ as a far witness.
Thus $y_i\parallel_P x_k$.
By \ref{cond:3}, $|S|=k\le m(w-1)$.

For \ref{desc:T}, let $S'=\{z_1>_P\ldots>_Pz_k\}$ be a descending chain in $S$, and set $P'=P[D_P(z_1)\cap U]$.
If $B$ is an antichain in $P'$ then $(A-D(z_1))\cup B$ is an antichain in $P$.
As $z_1$ is special, $|B| \le w/2$.
So
\begin{equation}
 \width( P') \le w/2.
\end{equation}
For  $i \in [k]$, let $w_i$ be the near witness for color $\gr{g}(z_i)$.
Note that $w_i \in P'$ since $w_i\leq z_i\leq z_1$.
By Dilworth's Theorem, there is a chain $T\subseteq\{w_1,\ldots,w_k\}$ with
$k \le \frac{w}{2} |T|$.
Each $w_i$ has different Groundy color, thus
by Proposition~\ref{claim:erdosz-szekeres}, it suffices to prove:
\begin{enumerate}[resume*]
 \item\label{desc:X} the size of any descending sequence in $T$ is at most $m(w-1)$,
 \item\label{asc:X} the size of any ascending sequence in $T$ is at most $\val_{\FF}(L_m, \lfloor \frac{w}{2} \rfloor )$.
\end{enumerate}

For \ref{desc:X}, let $s_1>_P \ldots>_P s_l$ be a descending chain in $T$.
For $1\le i\le l-1$, pick  a $\gr{g}(s_i)$-witness $t_i$  of $s_{i+1}$.
Using Proposition~\ref{claim:ascending_ladders},  \ref{cond:2} implies
$t_i \not<_P s_l$. Suppose $s_l <_P t_i$.
Then $t_i \in U$ and $t_i$ is also special as $t_i<_Ps_i$ by \ref{cond:1}.
As $\gr{g}(t_i)=\gr g(s_i)$ this contradicts the choice of $s_i$ as a near witness.
Thus $t_i\parallel_P s_l$.
By \ref{cond:3}, $|T|=l\le m(w-1)$.

For \ref{asc:X}, let $u_1 <_P \ldots <_P u_l$ be an ascending chain in $T$,
and for $i \in [l]$ let $v_i\in S'$ be the far $\gr{g}(u_i)$-witness.
Then $u_{l} \le_P v_{l}$.
As $S'$ is descending, $u_1 <_P \dots <_P u_l \le_P v_l  <_P \dots <_P v_1$.
Set $U_i = [u_i,v_i] \cap C_{\gr{g}(u_i)}$, $U' = \bigcup_{i=1}^{l} U_i$ and $P''=P'[U']$.
Define $\gr g': U' \rightarrow [l]$ by $\gr{g}'(x) = i \text{ iff } x \in U_i $.
Then $\gr{g}'$ is an $l$-Grundy coloring of $P''$: as $\gr g$ is a Grundy coloring,  if $i<j$ and $y\in U_j$ then there is $x\in C_i$ with $x\parallel_P y$; as $u_i<_Pu_j\le_P y\le_P v_j<_Pu_i$, we have $x\in U_i$.
Since $P''\subset P$ is $L_m$-free, $l \le \val_{\FF}(L_m, \lfloor w/2 \rfloor )$.
\end{proof}

Consider a Dilworth chain decompositions of $R$ and let $S$ be a chain with a maximum size in this decomposition.
Since the width of $R$ is at most $w$, we have
\[
  N-1 = |R| \leq w|S| = w|S\cap D| + w |S\cap U|.
\]
After applying Claim~\ref{maxint} we get
\[
 N \leq 1 + m^2w^2(w-1)^2\val_{\FF}(L_m, \lfloor w/2\rfloor ) \leq m^2w^4\val_{\FF}(L_m, \lfloor w/2\rfloor ).
\]
The equation~\eqref{nbound} with $n = \val_{\FF} (L_m, w)$ can be now rewrite into the following recursion
%
\begin{flalign*}
\val_{\FF} (L_m, w) &\le m^2 w^4 \val_{\FF} \left( L_m, \left\lfloor w/2 \right\rfloor \right)+\val_{\FF} (L_m,w-1).
\end{flalign*}
Applying this recursion repeatedly to the second term, with $\val_{\FF}(L_m,1)=1$, we obtain
\begin{flalign*}
\val_{\FF} (L_m, w) &
\le 1+\sum_{2 \le k \le w} m^2 k^4 \val_{\FF} \left( L_m, \left\lfloor k/ 2 \right\rfloor \right)
\le w m^2 w^4 \val_{\FF} \left(L_m, \left\lfloor w/2 \right\rfloor \right).
\end{flalign*}
Arguing by induction yields:
\begin{flalign*}
 \val_{\FF} (L_m, w) &\le m^{2 \lg w} w^{2.5\lg (2w) }\le  w^{2.5 \lg (2w) + 2 \lg m },
\end{flalign*}
which completes the proof of the lemma.
\end{proof}



\section{Limitations of Our Methods}
\label{sec:limitations}

Loosely speaking, two major parts of the proof of our main theorem rely on limiting the number of rungs in a ladder within a regular poset and the performance of First-Fit on the family $\Forb (L_m)$.
Here, we will show that our general upper bound for the online coloring problem cannot be greatly improved with our current methods.

In the first part of the section we show that the assertion of Lemma \ref{eladder} can not be improved.
Although $L_{2w^2+1}$ is not a subposet of any width $w$ regular poset, we show that there are regular posets of width $w$ that contain ladders whose number of rungs is quadratic in $w$.

\begin{lemma} \label{regladderlower}
For each integer $w \ge 2$, there is a regular poset $P^{\ll}$ so that
$\width(P)=w$ and $P$ contains $L_{w \lfloor (w+2) / 2 \rfloor }$ as an induced subposet.
\end{lemma}
\begin{proof}
Consider two antichains $A=\{u_1,\ldots, u_w\}$ and $B=\{v_1,\ldots,v_w\}$,
where $u_i$ and $v_i$ are the $i$-th elements of $A$ and $B$, respectively.
We say $(A,B,\le)$ is a core of:
\begin{itemize}
 \item \emph{type $I$} if for all $i,j \in [w]$
 $$u_i \le v_j \text{ iff } i=j,$$
 \item \emph{type $S_k$} for $k \in [w]$ if for all $i,j \in [w]$
 $$u_i \le v_j \text{ iff $i=j$ or ($i=1$ and $j \in [k]$) or ($i \in [2,k]$ and $j \in [i-1,i]$)},$$
 \item \emph{type $T_k$} for $k \in [w]$ if for all $i,j \in [w]$
 $$\text{$u_i \le v_j$ iff $i=j$ or  ($i \in [w-k+1,w]$ and $j=w$) or ($i \in [w-k]$ and $j \in [i-1,i]$).}$$
\end{itemize}
It is straightforward to verify that bipartite posets of types $I$, $S_k$ and $T_k$ are cores.
See Figure~\ref{figure:types_of_cores} for examples.

\begin{figure}[tbh]
\begin{center}
\begin{tikzpicture}
\path (0,5) coordinate (R3);
\path (0,2.5) coordinate (R2);
\path (0,0) coordinate (R1);
\foreach \i in {1,...,6}{
\path (R3) ++(\i,1) coordinate (y3\i);
\fill (y3\i) circle (3pt) node[above]{$v_\i$};
\path (R3) ++(\i,0) coordinate (x3\i);
\fill (x3\i) circle (3pt) node[below]{$u_\i$};}

\foreach \i in {1,...,6}{
\draw (x3\i) -- (y3\i);}

\path (x31) ++ (-3/2,1/2)node{$I = S_1 = T_1$};
\path (x36) ++ (3/2,1/2)node{$\phantom{I = S_1 = T_1}$};

\foreach \i in {1,...,6}{
\path (R2) ++(\i,1) coordinate (y2\i);
\fill (y2\i) circle (3pt) node[above]{$v_\i$};
\path (R2) ++(\i,0) coordinate (x2\i);
\fill (x2\i) circle (3pt) node[below]{$u_\i$};}

\foreach \i in {2,...,6}{
\draw (x2\i) -- (y2\i);
\draw (x21) -- (y2\i);}
\draw (x21) -- (y21);
\draw (x22) -- (y21);
\draw (x23) -- (y22);
\draw (x24) -- (y23);
\draw (x25) -- (y24);
\draw (x26) -- (y25);

\path (x21) ++ (-1,1/2)node{$S_6$};
\path (x26) ++ (1,1/2)node{$\phantom{S_6}$};

\foreach \i in {1,...,6}{
\path (R1) ++(\i,1) coordinate (y1\i);
\fill (y1\i) circle (3pt) node[above]{$v_\i$};
\path (R1) ++(\i,0) coordinate (x1\i);
\fill (x1\i) circle (3pt) node[below]{$u_\i$};}

\foreach \i in {1,...,6}{
\draw (x1\i) -- (y1\i);}
\foreach \i in {3,...,5}{
\draw (y16) -- (x1\i);}
\draw (x14) -- (y13);
\draw (x15) -- (y14);
\draw (x16) -- (y15);

\path (x11) ++ (-1,1/2)node{$T_4$};
\path (x16) ++ (1,1/2)node{$\phantom{T_4}$};
\end{tikzpicture}
\end{center}
\caption{Hasse diagrams of $I$, $S_6$, and $T_4$ for $w=6$.}\label{figure:types_of_cores}
\end{figure}
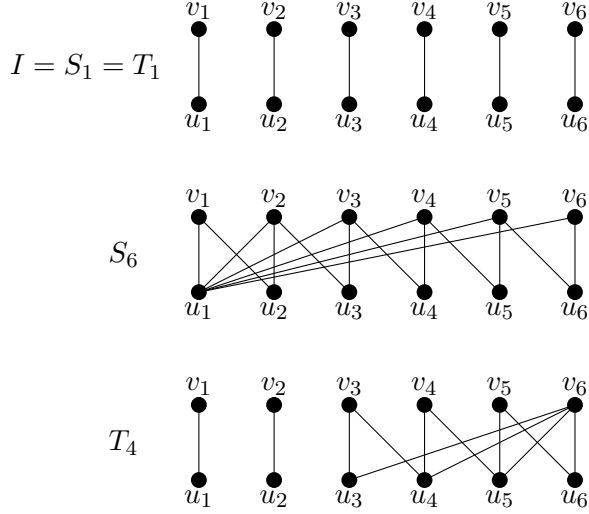

Now we construct an auxiliary regular poset $Q^{\ll}$, based on which the regular poset $P^{\ll}$ is built.
Let $V = \bigcup_{i=1}^{2w+1} A_i$, and let the presentation order of $Q^{\ll}$ be $A_1 \ll A_2 \ll \ldots \ll A_{2w+1}$.
The poset $Q=(V,\leq_Q)$ is defined as follows.
At every step in the presentation of $Q$, every two $\sqsubseteq_Q$-consecutive antichains
induce a core, one of type: $I$, $S_k$ or $T_k$ for $k \in [w]$.
Suppose that the antichains $A_{s(i)}$ and $A_{p(i)}$, if exist, denote the antichains
that are respectively just above and just below $A_i$ at the moment $A_i$ is presented.
To define the relation $\le_Q$ in $Q$ we need only to determine the relation $\le_Q$ between
$A_i$ and $A_{s(i)}$ and between $A_{p(i)}$ and $A_i$ at the moment $A_i$ is presented; the other comparabilities will follow by transitivity -- see \ref{R:transitivity}.
Below are the rules how $\le_Q$ is determined for the successively presented antichains $A_1,\ldots,A_{2w+1}$:
\begin{enumerate}
\item $A_2$ is set so that
\begin{itemize}
\item $s(2)=1$ and $(A_2,A_1,\le_Q)$ is of type $S_w$,
\item $p(2)$ is not defined.
\end{itemize}
\item For $i \in [3,w+1]$ the antichain $A_i$ is set so that 
\begin{itemize}
\item $s(i)=1$ and $(A_{i},A_{1},\le_Q)$ is of type $S_{w-i+2}$,
\item $p(i)=i-1$ and $(A_{i-1},A_{i}, \le_Q)$ is of type $I$.
\end{itemize}
\item $A_{w+2}$ is set so that
\begin{itemize}
\item $s(w+2)$ is not defined,
\item $p(w+2)=1$ and $( A_1, A_{w+2},\le_Q)$ is of type $T_w$.
\end{itemize}
\item For $i \in [w+3,2w+1]$ the antichain $A_i$ is set so that 
\begin{itemize}
\item $s(i)=i-1$ and $(A_{i},A_{i-1},\le_Q)$ is of type $I$,
\item $p(i)=1$ and $(A_{1},A_{i},\le_Q)$ is of type $T_{2w-i+2}$.
\end{itemize}
\end{enumerate}
The above rules imply the following relations between the antichains $A_1,\ldots,A_{2w+1}$ in the poset $Q$
(see Figure \ref{figure:reg_poset_Q}):
$$A_2 \sqsubset_Q A_3 \sqsubset_Q \ldots \sqsubset_Q A_{w+1} \sqsubset_Q A_1 \sqsubset_Q A_{2w+1} \sqsubset_Q A_{2w} \sqsubset_Q \ldots \sqsubset_Q A_{w+2}.$$
Although it is tedious to verify that $Q^{\ll}$ is indeed a width $w$ regular poset, it is straightforward and we leave it to the reader.

\noindent Let $\bot = A_2$, $\top=A_{w+2}$.
For every $i \in [w]$ we denote by:
\begin{itemize}
 \item $x_i$ -- the first point in $A_{i+1}$,
 \item $y_i$ -- the $w$-th point in $A_{2w+2-i}$,
 \item $b_i$ -- the $i$-th point in $\bot$,
 \item $t_i$ -- the $i$-th point in $\top$,
\end{itemize}
and finally we let $X=\{x_1,\ldots,x_w\}$ and $Y = \{y_1,\ldots,y_w\}$. 
By inspection we may easily check the following properties of $Q$.
\begin{enumerate}[label=(P\theenumi), series=Pconditions]
 \item \label{P:chains} $x_1 <_Q \ldots <_Q x_w$ and $y_1 <_Q \ldots <_Q y_w$.
\end{enumerate}
Moreover, for any $i,j \in [w]$:
\begin{enumerate}[resume*=Pconditions]
 \item \label{P:ladders} If $i \le j$ then $x_i <_Q y_j$, otherwise $x_i \parallel_Q y_j$.
 \item \label{P:tops_and_bottoms} If $j \le i \le j+2$ or $i=1$ or $j=w$ then $b_i \le_Q t_j$, otherwise $b_i \parallel_Q t_j$.
 \item \label{P:Y} If $j=w$ then $y_i <_Q t_j$, otherwise $y_i \parallel_Q t_j$.
 \item \label{P:X} If $i=1$ then $b_i <_Q x_j$, otherwise $b_i \parallel_Q x_j$.
\end{enumerate}

Now, we are ready to describe the regular poset $P^{\ll}$.
The poset $P$ will consists of $h = \lfloor (w+2)/2 \rfloor$ copies of $Q$.
We will use the same variable names to denote elements (sets) in the copies of $Q$ in $P$ as these introduced \mbox{for $Q$;}
however, we add the superscript $i$ to specify that a variable describes an element (a set) from the $i$-th copy of $Q$.
Formally, the poset $P=(V,\le_P)$ is defined such that $V = \bigcup_{i=1}^{h} V^{i}$ and $\le_P$ is the transitive closure of
$$(\le_{Q^1} \cup \ldots \cup \le_{Q^h}) \cup \{(t^{j}_i,b^{j+1}_i): i \in [w], j\in [h-1]\}.$$
The presentation order $\ll$ of $P$ is set so as:
\begin{enumerate}[label=(\roman*)]
\item\label{first} $V^{i} \ll V^{j}$ for any $1\leq i<j \leq h$,
\item\label{second} the order of the elements within every copy of $Q$ is the same as in $Q$.
\end{enumerate}
Again, checking that $P^{\ll}$ is a regular poset of width $w$ is straightforward;
an example of $P^{\ll}$ is shown in Figure~\ref{figure:reg_poset_Q}.

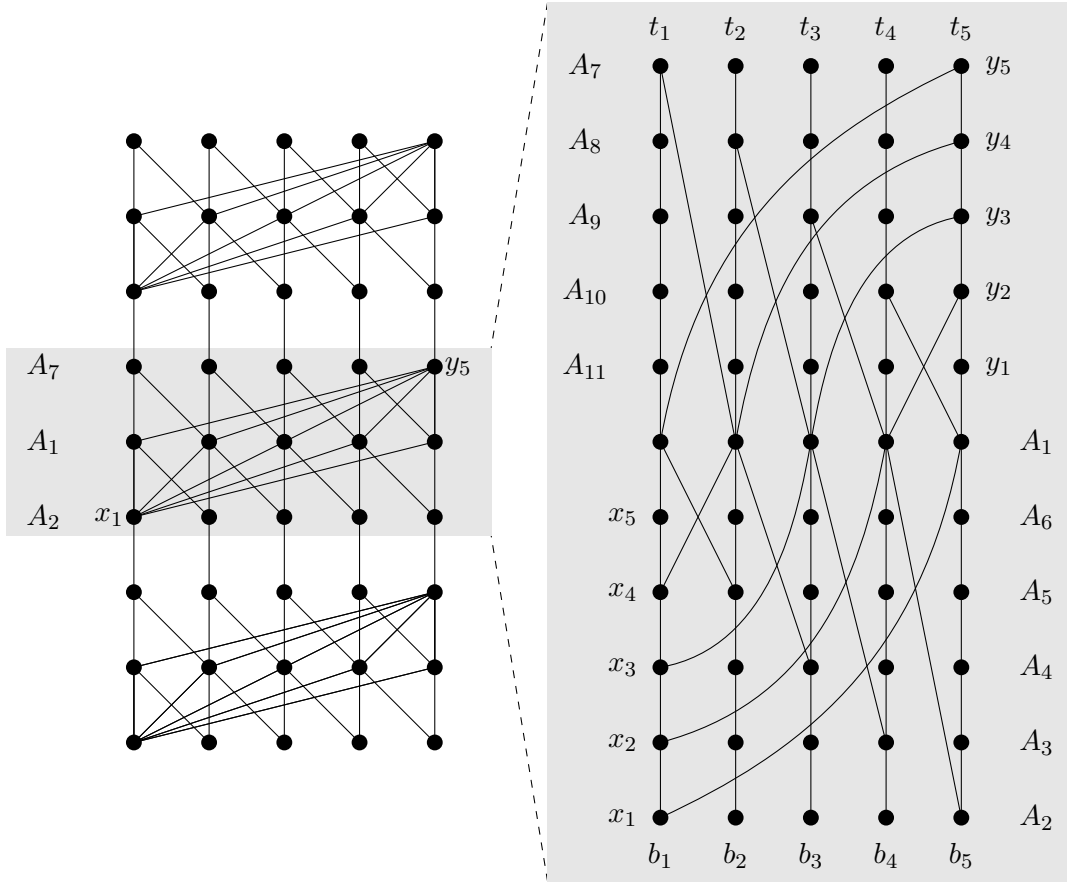
\begin{figure}[htb]
\begin{center}
\begin{tikzpicture}
\draw[dashed] (-.5,-.85) -- (-1.25,3.75);
\draw[dashed] (-.5,10.85) -- (-1.25,6.25);
\fill[gray!20] (6.5,-.85) -- (-.5,-.85) --  (-.5,10.85) -- (6.5,10.85) -- cycle;
\fill[gray!20] (-1.25,3.75) -- (-7.7,3.75) -- (-7.7,6.25) -- (-1.25,6.25) -- cycle;

\path (-7,1) coordinate (R0);
\foreach \i in {1,...,5}{
\path (R0) ++(\i,8) coordinate (j\i);
\fill (j\i) circle (3pt);
\path (R0) ++(\i,7) coordinate (h\i);
\fill (h\i) circle (3pt);
\path (R0) ++(\i,6) coordinate (g\i);
\fill (g\i) circle (3pt);
\path (R0) ++(\i,5) coordinate (f\i);
\fill (f\i) circle (3pt);
\path (R0) ++(\i,4) coordinate (e\i);
\fill (e\i) circle (3pt);
\path (R0) ++(\i,3) coordinate (d\i);
\fill (d\i) circle (3pt);
\path (R0) ++(\i,2) coordinate (c\i);
\fill (c\i) circle (3pt);
\path (R0) ++(\i,1) coordinate (b\i);
\fill (b\i) circle (3pt);
\path (R0) ++(\i,0) coordinate (a\i);
\fill (a\i) circle (3pt);}

\foreach \i in {1,...,5}{
\draw (a\i) -- (j\i);}
\foreach \i in {1,...,5}{
\draw (c5) -- (b\i);
\draw (a1) -- (b\i);}

\foreach \i in {1,...,5}{
\draw (c5) -- (b\i);
\draw (a1) -- (b\i);}

\foreach \i in {1,...,5}{
\draw (f5) -- (e\i);
\draw (d1) -- (e\i);}

\foreach \i in {1,...,5}{
\draw (j5) -- (h\i);
\draw (g1) -- (h\i);}

\draw (h2) -- (j1);
\draw (h3) -- (j2);
\draw (h4) -- (j3);
\draw (h5) -- (j4);

\draw (g2) -- (h1);
\draw (g3) -- (h2);
\draw (g4) -- (h3);
\draw (g5) -- (h4);

\draw (e2) -- (f1);
\draw (e3) -- (f2);
\draw (e4) -- (f3);
\draw (e5) -- (f4);

\draw (d2) -- (e1);
\draw (d3) -- (e2);
\draw (d4) -- (e3);
\draw (d5) -- (e4);

\draw (b2) -- (c1);
\draw (b3) -- (c2);
\draw (b4) -- (c3);
\draw (b5) -- (c4);

\draw (a2) -- (b1);
\draw (a3) -- (b2);
\draw (a4) -- (b3);
\draw (a5) -- (b4);

\path (f1) ++ (-1.2,0)node{$A_{7}$};
\path (f5)node[right]{$y_{5}$};
\path (e1) ++ (-1.2,0)node{$A_{1}$};
\path (d1) ++ (-1.2,0)node{$A_{2}$};
\path (d1)node[left]{$x_{1}$};

\path (0,0) coordinate (R1);

\foreach \i in {1,...,5}{
\path (R1) ++(\i,10) coordinate (a11\i);
\fill (a11\i) circle (3pt);}
\foreach \i in {1,...,5}{
\path (R1) ++(\i,9) coordinate (a10\i);
\fill (a10\i) circle (3pt);}
\foreach \i in {1,...,5}{
\path (R1) ++(\i,8) coordinate (a9\i);
\fill (a9\i) circle (3pt);}
\foreach \i in {1,...,5}{
\path (R1) ++(\i,7) coordinate (a8\i);
\fill (a8\i) circle (3pt);}
\foreach \i in {1,...,5}{
\path (R1) ++(\i,6) coordinate (a7\i);
\fill (a7\i) circle (3pt);}
\foreach \i in {1,...,5}{
\path (R1) ++(\i,5) coordinate (a6\i);
\fill (a6\i) circle (3pt);}
\foreach \i in {1,...,5}{
\path (R1) ++(\i,4) coordinate (a5\i);
\fill (a5\i) circle (3pt);}
\foreach \i in {1,...,5}{
\path (R1) ++(\i,3) coordinate (a4\i);
\fill (a4\i) circle (3pt);}
\foreach \i in {1,...,5}{
\path (R1) ++(\i,2) coordinate (a3\i);
\fill (a3\i) circle (3pt);}
\foreach \i in {1,...,5}{
\path (R1) ++(\i,1) coordinate (a2\i);
\fill (a2\i) circle (3pt);}
\foreach \i in {1,...,5}{
\path (R1) ++(\i,0) coordinate (a1\i);
\fill (a1\i) circle (3pt);}

\draw (a11) -- (a111);
\draw (a12) -- (a112);
\draw (a13) -- (a113);
\draw (a14) -- (a114);
\draw (a15) -- (a115);

\draw (a15) --(a64);
\draw (a24) --(a63);
\draw (a33) --(a62);
\draw (a42) --(a61);

\draw (a84) --(a65);
\draw (a93) --(a64);
\draw (a102) --(a63);
\draw (a111) --(a62);

\draw (a115) .. controls +(205:2) and +(80:3) .. (a61);
\draw (a105) .. controls +(195:2) and +(80:2) .. (a62);
\draw (a95) .. controls +(190:1.5) and +(80:1) .. (a63);
\draw (a85) --(a64);

\draw (a41) --(a62);
\draw (a31) .. controls +(10:1.5) and +(260:1) .. (a63);
\draw (a21) .. controls +(15:2) and +(260:2) .. (a64);
\draw (a11) .. controls +(25:2) and +(260:3) .. (a65);

\path (a111) ++ (0,1/2)node{$t_{1}$};
\path (a112) ++ (0,1/2)node{$t_{2}$};
\path (a113) ++ (0,1/2)node{$t_{3}$};
\path (a114) ++ (0,1/2)node{$t_{4}$};
\path (a115) ++ (0,1/2)node{$t_{5}$};

\path (a111) ++ (-1,0)node{$A_{7}$};
\path (a115) ++ (1/2,0)node{$y_{5}$};

\path (a101) ++ (-1,0)node{$A_{8}$};
\path (a105) ++ (1/2,0)node{$y_4$};

\path (a91) ++ (-1,0)node{$A_{9}$};
\path (a95) ++ (1/2,0)node{$y_3$};

\path (a81) ++ (-1,0)node{$A_{10}$};
\path (a85) ++ (1/2,0)node{$y_2$};

\path (a71) ++ (-1,0)node{$A_{11}$};
\path (a75) ++ (1/2,0)node{$y_1$};

\path (a65) ++ (1,0)node{$A_{1}$};

\path (a51) ++ (-1/2,0)node{$x_{5}$};
\path (a55) ++ (1,0)node{$A_{6}$};

\path (a41) ++ (-1/2,0)node{$x_4$};
\path (a45) ++ (1,0)node{$A_{5}$};

\path (a31) ++ (-1/2,0)node{$x_3$};
\path (a35) ++ (1,0)node{$A_{4}$};

\path (a21) ++ (-1/2,0)node{$x_2$};
\path (a25) ++ (1,0)node{$A_{3}$};

\path (a11) ++ (-1/2,0)node{$x_1$};
\path (a15) ++ (1,0)node{$A_{2}$};

\path (a11) ++ (0,-1/2)node{$b_1$};
\path (a12) ++ (0,-1/2)node{$b_2$};
\path (a13) ++ (0,-1/2)node{$b_3$};
\path (a14) ++ (0,-1/2)node{$b_4$};
\path (a15) ++ (0,-1/2)node{$b_5$};

\end{tikzpicture}
\end{center}
\caption{
The width $5$ poset $Q$ is shown on the right.
The sketch of the construction of the width $5$ poset $P$ is shown on the left. 
It consists of $3$ copies of $Q$ (the middle copy of $Q$ in $P$ is depicted with gray background) joined as shown in the figure. 
}\label{figure:reg_poset_Q}
\end{figure}

To finish the proof of the lemma we show that
\begin{equation}
\label{eq:ladder_in_P}
\text{the set } \bigcup_{j=1}^{h} (X^{j} \cup Y^{j}) \text{ induces an $(w\cdot h)$-ladder in $P$,}
\end{equation}
with $x_i^{j}y_i^{j}$ being its $((j-1)h+i)$-th rung. Clearly, we have
\begin{equation}
\label{eq:chains_X_Y}
X^{1} <_P \ldots <_P X^{h} \text{ and } Y^1 <_P \ldots <_P Y^h
\end{equation}
by the definition of $\le_P$. Finally, we will show that for all $i,j \in [h]$:
\begin{equation}
\label{eq:sets_X_Y}
X^i <_P Y^j \text{ if } i<j \text{ and } X^i \parallel_Q Y^j \text{ if } i > j.
\end{equation}
Note that the relation between $X^{i}$ and $Y^{j}$ in the case when $i=j$ is handled by \ref{P:ladders}.
Clearly, if we prove \eqref{eq:sets_X_Y}, \eqref{eq:ladder_in_P} follows by \ref{P:chains}, \eqref{eq:chains_X_Y}, \eqref{eq:sets_X_Y}, and \ref{P:ladders}.
Assume that $i < j$.
Clearly, by \ref{P:ladders} it follows that $X^i$ is less than the greatest element in $Y^i$.
Consequently, $X^i <_P Y^j$ by \eqref{eq:chains_X_Y}.
Assume $i > j$.
We consider only the case $i=h$ and $j=1$; the remaining ones are even easier to prove.
First note that every comparability between a point in $Y_1$ and a point in $X_h$ needs to be
implied by transitivity on some point from $\bot^h$.
Note that $D_P(X_h) \cap \bot^h$ contains only
the first element of $\bot^h$ by \ref{P:X}.
By \ref{P:tops_and_bottoms} and \ref{P:Y}, note that the set
$U_P(Y_1) \cap A^{i}_2$ contains exactly $2i-3$ last elements in
$\bot^i$ for $i \in [2,h]$. Pluging $h = \lfloor (w+2)/2 \rfloor$ to the last observation we get
$U_P(Y_1) \cap \bot^{h}$ contains not more than $2 \lfloor(w+2)/2 \rfloor -3 \le w-1$ last elements from $\bot^h$.
In particular, $U_P(Y_1) \cap \bot^{h}$ does not contain the first element of $\bot^{h}$.
It follows that $X^h \parallel_P Y^1$.
\end{proof}

In the last part of this section we give the lower bound on $\val_{\FF}(L_m,w)$, which shows that the upper bound from Lemma \ref{ladupbound} 
can not be substantially improved.
For the upcoming construction we remind the definition of the \emph{lexicographical product} of two posets.
For posets $P$ and $Q$, the lexicographical product $P \cdot Q$ is the poset
with vertices $\{ (p,q) : p \in P, q \in Q\}$ and order $\le_{P \cdot Q}$, where
$$(p_1,q_1) \le_{P \cdot Q} (p_2,q_2) \text{ if either } p_1 <_P p_2 \text{ or } (p_1 = p_2 \text{ and }q_1 \le_Q q_2).$$
Informally, we may think of $P \cdot Q$ as the poset $P$ where each vertex has been ``inflated'' to a copy of $Q$.
It is well know that
\begin{equation} \label{lexwidth}
\width(P\cdot Q) = \width(P)\width(Q).
\end{equation}
The following two simple properties (we left the proof for the reader) are the key in the proof of the upcoming lemma.
For $p,r \in P$ and $u,v,s \in Q$ we have:
\begin{align}
 &\textrm{If ($(p,u) \le_{P\cdot Q}(r,s)$ or $(p,u) \ge_{P\cdot Q}(r,s)$) and }(r,s) \parallel_{P \cdot Q} (p,v), \textrm{then } p=r. \label{comppar}\\
 &\textrm{If } (p,u) \le_{P \cdot Q} (r,s) \le_{P \cdot Q} (p,v), \textrm{then } p=r. \label{between}
\end{align}

\begin{lemma}
\label{laddnbound}
For $m, w \in \mathbb Z^+$ with $m >1$, we have  $w^{\lg (m-1)}/(m-1) \le \val_{\FF} (L_m, w).$
\end{lemma}
\begin{proof}
Fix $m \in \mathbb Z^+$ with $m>1$. Let $R$ be the width $2$ poset $R_{m-1}$ as defined in the proof of Lemma~\ref{CE}.
For technical reasons we would like $R$ to have the least and the greatest element.
Vertex $x_1^{m-1}$ is already the greatest in $R$, but there is no least element in $R$.
Therefore we extend $R$ to $P$ by adding a new element $\hat 0$ which is below entire $R$.
The greatest element in $P$ is still $x_1^{m-1}$, which we denote by $\hat 1$.

It is a simple exercise to see that $P$ also satisfies the statement of Lemma~\ref{CE},
i.e., $\width(P) = 2$ and $\chi_{\FF}(P) \ge \chi_{\FF}(R) \ge m-1$. As $R$ is an induced subposet of $P$ we have $I_P(\hat 0) = \emptyset$ and $\abs{I_P (x_i^k)} = k <  m-1 $ for  $1\le i \le k < m-1$
and $\abs{I_P(x_i^{m-1})} = i-1<m-1$ for $i\in[m-1]$.
Observe that in a ladder $L_m$, the lowest vertex of the upper leg is always incomparable to $m-1$ vertices.
Hence, there is no vertex in $P$ that can serve as the lowest vertex of the upper leg of an $m$-ladder and thus
\begin{equation} \label{PForb}
P \in \Forb(L_m).
\end{equation}

We are prepared to build a poset $Q_k \in \Forb(L_m)$ with $n$-Grundy coloring so that $\width(Q_k) = 2^k$ and $n \ge (m-1)^k$.
Poset $Q_k$ is defined by the following rules:
\begin{enumerate}[label=(Q\theenumi)]
\item $Q_0$ is a single vertex $z$.
\item $Q_{k+1} = P \cdot Q_{k}$.
\end{enumerate}
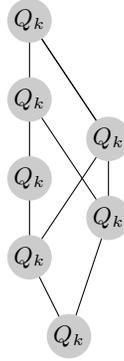
\begin{figure}[tbh]
\begin{center}
\begin{tikzpicture}[scale=.7]
\path (0,0) coordinate (R1);

\path (R1) ++(0,6) coordinate (x7);
\path (R1) ++(0,4.5) coordinate (x6);
\path (R1) ++(0,3) coordinate (x5);
\path (R1) ++(1.5,3.75) coordinate (x4);
\path (R1) ++(1.5,2.25) coordinate (x3);
\path (R1) ++(0,1.5) coordinate (x2);
\path (R1) ++(.75,0) coordinate (x1);

\draw (x2) -- (x7);
\draw (x4) -- (x7);
\draw (x3) -- (x4);
\draw (x1) -- (x3);
\draw (x1) -- (x2);

\draw (x2) -- (x4);

\draw (x4) -- (x7);
\draw (x3) -- (x6);

\foreach \j in {1,...,7}{
\fill[gray!40] (x\j) circle (12pt) node[black]{\footnotesize $Q_{k}$};}

%
%
%
%
%
%
%

\end{tikzpicture}
\end{center}
\caption{Simplified Hasse diagram of  $Q_{k+1}$ with $m=4$.}\label{lexQ}
\end{figure}

Note that $Q_1$ and $P$ are isomorphic and so we will treat $Q_1$ as $P$.
The next two properties are the consequence of the definition of $Q_k$, equation~\eqref{lexwidth} and the fact that $P$ has the least and the greatest element with $\width(P)=2$. For each $k\in\mathbb N$
  \begin{enumerate}[resume*]
  \item\label{Qmaxmin} $Q_k$ has a minimum vertex and a maximum vertex,
  \item $\width (Q_k) = 2^k$.
 \end{enumerate}

\begin{claim}
For each $k \in \mathbb N$, $Q_k \in \Forb (L_m)$.
\end{claim}
\begin{proof}
We will use induction on $k$.
For our bases, we see $k=0$ is trivial and $k=1$ is established by~\eqref{PForb}.
Take $k > 1$ and suppose the inductive hypothesis holds for all smaller cases.
Assume $L$ is an $m$-ladder in $Q_k$ with
the lower leg $(a_1, u_1) <_{Q_k} (a_2 , u_2 ) <_{Q_k} \ldots <_{Q_k} (a_m, u_m)$
and the upper leg $(b_1, v_1) <_{Q_k} (b_2, v_2) <_{Q_k} \ldots <_{Q_k} (b_m, v_m)$.
If all vertices of $L$ are pairwise different in the first coordinate,
then these vertices would induce an $m$-ladder in $P$, which violates~\eqref{PForb}.
Hence, at least two vertices of $L$ share a first coordinate, say $p \in P$.
Let $Q' = \{ (p,q) : q \in Q_{k-1 } \}$ and note that $Q'$ and $Q_{k-1}$ are isomorphic.
Let $\hat 0$ and $\hat 1$ to be the minimum and the maximum, respectively, vertices of $Q'$ (which exist by~\ref{Qmaxmin}).

Assume for a while, $Q'$ contains two vertices of the lower leg of $L$, i.e.,
there are $i<j \in [m]$ so that $(a_i,u_i),(a_j,u_j) \in Q'$ with $a_i=a_j=p$.
From the definition of a ladder, we know $(a_i,u_i) \le_{Q_k} (b_i,v_i) \parallel_{Q_k} (a_j,u_j)$.
By~\eqref{comppar} we have $b_i=p$ and thus $Q'$ contains $(b_i,v_i)$, a vertex of the upper leg of $L$.
For similar reasons, if $Q'$ contains two vertices of the upper leg of $L$, then it has to have one of the lower leg of $L$.
Therefore, there are $(a_i,u_i),(b_j,v_j) \in Q'$, vertices of the lower and the upper leg of $L$, respectively.
We see $(a_i,u_i) \le_{Q_k} (a_m,u_m) \parallel_{Q_k} (b_j,v_j)$ (if $j<m$)  or $(a_i,u_i) \le_{Q_k} (a_m,u_m) \le_{Q_k} (b_j,v_j)$ (if $j =m$).
In the former case we use~\eqref{comppar} and in the latter case~\eqref{between} to show $(a_m,u_m) \in Q'$.
Similarly, $(a_i,u_i) \parallel_{Q_k} (b_1,v_1) \le_{Q_k} (b_j,v_j)$ (if $i > 1$) or $(a_i,u_i) \le_{Q_k} (b_1,v_1) \le_{Q_k} (b_j,v_j)$ (if $i=1$).
Again, using~\eqref{comppar} or~\eqref{between}, we have $(b_1,v_1) \in Q'$.

For any vertex $(r,s)$ in $L$ so that $(r,s) \notin \{ (a_1,u_1), (b_m,v_m) \}$,
we have either $(b_1,v_1) \le_{Q_k} (r,s) \parallel_{Q_k} (a_m,u_m)$ or $(b_1,v_1) \parallel_{Q_w} (r,s) \le_{Q_k} (a_m,u_m)$.
By~\eqref{comppar} we deduce $(r,s) \in Q'$.
Finally, the vertices
$$
\{ \hat 0 , (a_2,u_2) , (a_3,u_3) , \ldots , (a_m,u_m) , (b_1,v_1) , (b_2,v_2) , \ldots , (b_{m-1},v_{m-1}) , \hat 1 \} \subseteq Q'
$$
induce an $m$-ladder in $Q'$, which contradicts the inductive hypothesis, proving the claim.
\end{proof}

\begin{claim} \label{Precur}
$\chi_{\FF} (Q_{k+1}) \ge (m-1) \chi_{\FF} (Q_{k})$.
\end{claim}
\begin{proof}
We already know $P$ has an ($m-1$)-Grundy coloring, say $\gr{f}$.
Let $\gr{g}$ be a $n$-Grundy coloring of $Q_{k}$.
Define $\gr{h}: Q_{k+1} \rightarrow [(m-1)n]$ by $\gr{h}((p,q)) = (\gr{f}(p)-1)n+\gr{g}(q)$.
We will show $\gr{h}$ is an ($(m-1)n$)-Grundy coloring of $Q_{k+1}$.
For that we need to prove~\ref{G:1}-\ref{G:3} of Definition~\ref{Pgr}.

It is easy to check that a function $(f,g) \to (f-1)n+g$ is a bijection between $[m-1]\times[n]$ and $[(m-1)n]$.
Since $\gr{f}$ and $\gr{g}$ are surjective, then also $\gr h$ must be surjective.
Thus, $\gr h$ satisfies~\ref{G:2}.
To show~\ref{G:1} suppose $\gr{h}((p,q))=\gr{h}((r,s))$.
This implies that $\gr{f}(p) = \gr{f}(r)$ and $\gr{g}(q) = \gr{g}(s)$.
By~\ref{G:1} of $\gr f$ and $\gr g$,
two pairs of vertices $p,r$ and $q,s$ are comparable respectively in $P$ and in $Q_k$.
Therefore, by the definition of the lexicographical product,
vertices $(p,q)$ and  $(r,s)$ are comparable in $Q_{k+1}$ and condition~\ref{G:1} holds for $\gr{h}$.

Consider $(r,s)\in Q_{k+1}$ so that $\gr{h}((r,s))=j>1$ and take any $i<j$.
We will show $(r,s)$ has an $i$-witness in $Q_{k+1}$ which will prove~\ref{G:3}.
There are unique integers $c \in [m-1]$ and $d \in [n]$ so that $j=(c-1)n+d$ and $\gr{f}(r)=c$, $\gr{g}(s)=d$.
Similarly, we can find $a \in [m-1]$ and $b \in [k]$ so that $i=(a-1)n+b$.
As $i<j$, we must have $a \le c$.

Suppose $a=c$, then $b<d$.
As $\gr{g}$ satisfies~\ref{G:3}, there is some $q\in Q_{k}$ so that $\gr{g}(q)=b$ and $q \parallel_{Q_{k}} s$.
By the definition of lexicographical product, $(r,q) \parallel_{Q_{k+1}} (r,s)$.
Observe $ \gr{h} ((r,q)) =  i$ and then $(r,q)$ is the desired witness.

The case  $a<c$ is similar.
This time we use~\ref{G:3} of $\gr{f}$ to get $p \in P$ so that $\gr{f}(p)=a$ and $p \parallel_P r$.
Take any $q\in Q_k$ so that $\gr g (q) = b$ ($q$ exists by~\ref{G:2} of $\gr g$).
Again, by the definition of lexicographical product, $(p,q) \parallel_{Q_{k+1}} (r,s)$.
Finally, as $\gr{h}((p,q))=i$, we deduce $(p,q)$ is the desired witness in this case.
\end{proof}
Claim~\ref{Precur} with $\chi_{\FF} (Q_0) =1$ implies $\chi_{\FF}(Q_k) \ge (m-1)^k$.
Note that $\width(Q_{\lfloor \lg w \rfloor})$ could be less then $w$.
But we can always add some isolated vertices to $Q_{\lfloor \lg w \rfloor}$
to get width $w$ poset $Q'$ so that $\chi_{\FF}(Q') \ge \chi_{\FF} (Q_{\lfloor \lg w \rfloor})$.
This finally shows
$$
\val_{\FF} (L_m,w) \ge \chi_{\FF}(Q_{\lfloor \lg w \rfloor}) \ge (m-1)^{\lfloor \lg w \rfloor} \ge \frac{w^{\lg (m-1)}}{m-1}.
$$
\end{proof} 

Lemmas~\ref{regladderlower} and \ref{laddnbound} show that the upper bound of $\val (\mathcal P_w)$ cannot be pushed below $w^{\lg w}$ using our current methods.

\section{Concluding Remarks}
\label{sec:concluding}

Although we have improved the upper bound for $\val (\mathcal P_w)$, our current methods cannot bring it down to a polynomial bound without some major changes.
Perhaps improvements in the understanding of regular posets could lead us to a subfamily of more interesting forbidden substructures.
We could also examine online coloring algorithms other than First-Fit to reduce the number of colors used on the family $\Forb (L_m)$.

We may look beyond the scope of $\val (\mathcal P_w)$.
So far, the reduction to regular posets has only been studied on general posets.
We might ask what the results of Procedures (1) and (2) are when we start with a poset from $\Forb (Q)$ (for some poset $Q$).
It is interesting to ask what analogues of the inequality (\ref{eqM}) could be built.
For instance, could an analogue for cocomparability graphs be created?
Already, Kierstead, Penrice, and Trotter~\cite{KPT} have shown that a cocomparability graph can be
colored online using a bounded number of colors.
However, this bound is so large that it was not computed.
Perhaps methods similar to the reduction to regular posets could be created.

\end{document}